\documentclass[a4paper,UKenglish,cleveref,hyperref,autoref]{lipics-v2021}

\title{Mim-Width is paraNP-complete}
\titlerunning{Mim-Width is paraNP-complete}

\author{Benjamin Bergougnoux}{LIS, Aix-Marseille Université, France \and \url{https://benjaminbergougnoux.github.io/}}{benjamin.bergougnoux@gmail.com}{https://orcid.org/0000-0002-6270-3663}{}
\author{\'{E}douard Bonnet}{CNRS, ENS de Lyon, Universit\'e Claude Bernard Lyon 1, LIP UMR 5668, Lyon, France \and \url{http://perso.ens-lyon.fr/edouard.bonnet}}{edouard.bonnet@ens-lyon.fr}{https://orcid.org/0000-0002-1653-5822}{}
\author{Julien Duron}{ENS de Lyon, Universit\'e Claude Bernard Lyon 1, LIP UMR 5668, Lyon, France \and \url{http://perso.ens-lyon.fr/julien.duron}}{julien.duron@ens-lyon.fr}{https://orcid.org/0009-0004-0925-9438}{}

\authorrunning{B. Bergougnoux, \'E. Bonnet, and J. Duron}

\Copyright{Benjamin Bergougnoux, Édouard Bonnet, and Julien Duron}

\category{}

\relatedversion{}

\supplement{}

\funding{}

\acknowledgements{}

\nolinenumbers 

\hideLIPIcs  

\EventEditors{John Q. Open and Joan R. Access}
\EventNoEds{2}
\EventLongTitle{42nd Conference on Very Important Topics (CVIT 2016)}
\EventShortTitle{CVIT 2016}
\EventAcronym{CVIT}
\EventYear{2016}
\EventDate{December 24--27, 2016}
\EventLocation{Little Whinging, United Kingdom}
\EventLogo{}
\SeriesVolume{42}
\ArticleNo{23}

\usepackage[utf8]{inputenc}  

\usepackage[T1]{fontenc}
\usepackage{lmodern}
\usepackage[colorinlistoftodos,bordercolor=orange,backgroundcolor=orange!20,linecolor=orange,textsize=normalsize]{todonotes}
\usepackage{amsmath}  
\usepackage{amssymb}
\usepackage{bbm}
\usepackage{accents}
\usepackage{complexity}
\usepackage{booktabs}
\usepackage{bm}
\usepackage{paralist}
\usepackage{fixmath}
\usepackage{tcolorbox}

\makeatletter
\newtheorem*{rep@theorem}{\rep@title}
\newcommand{\newreptheorem}[2]{%
\newenvironment{rep#1}[1]{%
 \def\rep@title{#2 \ref{##1}}%
 \begin{rep@theorem}}%
 {\end{rep@theorem}}}
\makeatother

\newreptheorem{theorem}{Theorem}
\newreptheorem{lemma}{Lemma}

\usepackage{xspace}
\usepackage{tikz}
\usepackage[ruled,vlined,linesnumbered]{algorithm2e}

\usetikzlibrary{fit}
\usetikzlibrary{arrows}
\usetikzlibrary{patterns}
\usetikzlibrary{calc}
\usetikzlibrary{shapes}
\usetikzlibrary{positioning}
\usetikzlibrary{math}
\usetikzlibrary{shapes.symbols}
\usetikzlibrary{decorations.pathreplacing,calligraphy}
\usetikzlibrary{arrows.meta}
\usetikzlibrary{shapes.geometric}
\usepackage[scr=boondox,scrscaled=1.05]{mathalfa}

\renewcommand{\geq}{\geqslant}
\renewcommand{\leq}{\leqslant}

\renewcommand{\le}{\leq}
\renewcommand{\ge}{\geq}
\newcommand{\Nn}{\mathbb{N}}

\newcommand{\Aux}{\text{Aux}}

\newcommand{\defparproblem}[
4]{
  \vspace{1mm}
  \begin{tcolorbox}[
    colframe=black,        
    colback=white,         
    boxrule=0.5pt,         
    arc=4pt,               
    left=6pt, right=6pt,   
    top=6pt, bottom=6pt    
  ]
    \begin{tabular*}{\textwidth}{@{\extracolsep{\fill}}lr}
      #1 & {\bf{Parameter:}} #3 \\
    \end{tabular*} \\
    {\bf{Input:}} #2 \\
    {\bf{Question:}} #4
  \end{tcolorbox}
  \vspace{1mm}
}

\newcommand{\paraNP}{\ensuremath{\mathrm{paraNP}}}

\newenvironment{proofofclaim}{\noindent \textsc{Proof:}}{\hfill$\Diamond$\medskip}

\newcommand{\mim}{\text{mim}}
\newcommand{\ssim}{\text{sim}}
\newcommand{\omim}{\text{omim}}

\newcommand{\weight}{\omega}

\newcommand{\ldb}{\textsc{Linear Degree Balancing}\xspace}
\newcommand{\tdb}{\textsc{Tree Degree Balancing}\xspace}
\newcommand{\tmb}{\textsc{Tree Mim-Balancing}\xspace}
\newcommand{\tsb}{\textsc{Tree Sim-Balancing}\xspace}
\newcommand{\lmb}{\textsc{Linear Mim-Balancing}\xspace}
\newcommand{\lsb}{\textsc{Linear Sim-Balancing}\xspace}

\newcommand{\tmap}{tree mapping\xspace}
\newcommand{\tmaps}{tree mappings\xspace}
\newcommand{\pmap}{path mapping\xspace}
\newcommand{\pmaps}{paths mappings\xspace}
\newcommand{\htree}{hybrid tree\xspace}
\newcommand{\htrees}{hybrid trees\xspace}

\begin{document}

\maketitle

\begin{abstract}
  We show that it is \NP-hard to distinguish graphs of linear mim-width at~most 1211 from graphs of~sim-width at~least 1216.
  This implies that \textsc{Mim-Width}, \textsc{Sim-Width}, \textsc{One-Sided Mim-Width}, and their linear counterparts are all \paraNP-complete, i.e., \NP-complete to compute even when upper bounded by a~constant.
  A key intermediate problem that we introduce and show \NP-complete, \textsc{Linear Degree Balancing}, inputs an edge-weighted graph $G$ and an integer~$\tau$, and asks whether $V(G)$ can be linearly ordered such that every vertex of $G$ has weighted \emph{backward} and \emph{forward} degrees at~most~$\tau$.
\end{abstract}

\section{Introduction}\label{sec:intro}

While it was shown shortly after the inception of these parameters by Vatshelle in 2012~\cite{Vatshelle12,BelmonteV13} that \textsc{Mim-Width} and \textsc{Linear Mim-Width} are W[1]-hard~\cite{SaetherV15,SaetherV16}, whether a~slice-wise polynomial (XP) algorithm\footnote{i.e., for any fixed integer $k$, a~polynomial-time algorithm (whose exponent may depend on~$k$) that decides if the (linear) mim-width of the input graph is at~most~$k$.} can compute (or approximate) the (linear) mim-width of an input graph has been raised as an open question repeatedly over the past twelve years~\cite{Vatshelle12,SaetherV16,JaffkeKT20,JaffkeKST19,BergougnouxPT22,BergougnouxKR23,OtachiST24,BergougnouxDJ23}.
We give a~negative answer to this question (at least for some too-good approximation factor), and similarly settle the parameterized complexity of the related sim-width and one-sided mim-width parameters, as well as their linear variants.
Indeed we show that all these parameters are \paraNP-complete to compute, i.e., \NP-complete even when guaranteed to be upper bounded by a~universal constant.

\begin{theorem}\label{thm:main-csq}
  \textsc{Mim-Width}, \textsc{Sim-Width}, \textsc{One-Sided Mim-Width}, \textsc{Linear Mim-Width}, \textsc{Linear Sim-Width}, and \textsc{Linear One-Sided Mim-Width} are \paraNP-complete.
\end{theorem}

We show~\cref{thm:main-csq} with a~single reduction.

\begin{theorem}\label{thm:main}
  There is a~polynomial-time algorithm that takes an input $\varphi$ of \textsc{4-Occ Not-All-Equal 3-Sat} and builds a~graph $G^*$ such that
  \begin{compactitem}
    \item if $\varphi$ is satisfiable, then $G^*$ has linear mim-width at~most~1211,
    \item if $\varphi$ is unsatisfiable, then $G^*$ has sim-width at~least~1216.
  \end{compactitem}
\end{theorem}

\Cref{thm:main} indeed implies~\cref{thm:main-csq} as the linear mim-width upper bounds the other five parameters, while the sim-width lower bounds the other five parameters.
Our reduction is naturally split into three parts, thereby going through two intermediate problems.
The first intermediate problem may be of independent interest (perhaps especially so, its unweighted version), and we were somewhat surprised not to find it already defined in the literature.
We call it \ldb.

\defparproblem{\ldb}{An edge-weighted $n$-vertex graph~$H$ and a~non-negative integer~$\tau$.}{$\tau$}{Is there a~linear ordering $v_1 \prec v_2 \prec \ldots \prec v_n$ of $V(H)$ such that every vertex~$v_i$ has weighted degree in $H[\{v_1, \ldots, v_i\}]$ and in $H[\{v_i, \ldots, v_n\}]$ at~most~$\tau$?}

We call \emph{$\tau$-balancing order} of~$H$ a~linear order over $V(H)$ witnessing that $H$ is a~positive instance of~\ldb.

\medskip

\textbf{The three steps.}
Our reduction starts with a~\textsc{Not-All-Equal 3-Sat} (\textsc{Nae 3-Sat} for short) instance $\varphi$, and goes through an edge-weighted graph $(H, \weight)$, a~vertex-partitioned graph $(G,\mathcal P)$, and finally an instance $G^*$ of \textsc{Mim-Width}.

First we prove that \ldb is \NP-complete even when $\tau$ is a~constant, and every edge weight is a~positive integer.
For our purpose, we in fact show something stronger.
The first step is a~polynomial-time reduction from \textsc{4-Occ Nae 3-Sat} that maps satisfiable formulas to edge-weighted graphs admitting a~$\tau$-balancing order, and unsatisfiable formulas to negative instances of \tdb, a~tree variant of \ldb, for the larger threshold of $\tau + \gamma$, where $\gamma$ can grow linearly in~$\tau$.

In~\tdb, the vertices of~$H$ are bijectively mapped to the nodes of a~freely-chosen tree $T$ such that for every node $t$ of~$T$ and every edge $e$ incident to $t$, the vertex of $H$ mapped to $t$ has weighted degree at~most the given threshold in the cut of~$H$ defined by the two connected components of~$T-e$. 
A~formal definition is given in~\cref{sec:lin-ord-to-trees}.

The second step turns the weighted degree into the maximum (semi-)induced matching at the expense of mapping \emph{subsets} of vertices of $G$ to nodes of~$T$, in a~way that the nodes of~$T$ jointly hold a~prescribed partition $\mathcal P$ of~$V(G)$.
In~\tmb (resp.~\tsb), for every edge~$e$ of~$T$, the size of a~maximum semi-induced (resp.~induced) matching in the cut of $G$ defined by~$e$ shall remain below the threshold.
Their linear variants force $T$ to be a~path.
See~\cref{sec:im-balancing} for formal definitions.

The third step erases the differences between \tmb and \textsc{Mim-Width}, and between their respective variants.
Intuitively speaking:
\begin{compactitem}
\item for each part of~$\mathcal P$, the corresponding vertices of $G^*$ can be gathered in their own subtree,
\item $T$ can be chosen ternary (i.e., every non leaf node has degree~3),
\item only the leaves of~$T$ need hold a~vertex of~$G^*$.
\end{compactitem}

\Cref{fig:roadmap} summarizes these three steps.
\begin{figure}[h!]
  \centering
  \resizebox{\textwidth}{!}{
  \begin{tikzpicture}[arr/.style={thick, -stealth}]
    \def\s{3.5}
    \def\v{0.8}

    \foreach \i/\j/\c in {-0.35/0.35/blue, 0.35/1.36/purple, 1.36/2.6/yellow, 2.6/3.8/orange}{
      \fill[opacity=0.15, color=\c] (\s * \i, -1.3 * \v) -- (\s * \i,2.5 * \v) -- (\s * \j,2.5 * \v) -- (\s * \j, -1.3 * \v) -- cycle ;
    }

    \foreach \i/\p/\x/\sat/\unsat/\inp in {
      1/\textsc{Nae 3-Sat}/0/satisfiable/unsatisfiable/{$\varphi$},
      2/\textsc{Degree Balancing}/0.87/$\tau$-balancing order/no $(\tau+\gamma)$-balancing tree/{$(H, \weight)$},
      3/\textsc{$\bullet$im-Balancing}/1.95/linear mim-balancing $\leqslant \tau+50$/sim-balancing $> \tau + \gamma$/{$(G,\mathcal P)$},
      4/\textsc{$\bullet$im-Width}/3.2/lin. mim-width $\leqslant \frac{46}{45}\tau +107$/sim-width $> \tau+\gamma$/{$G^*$}}
             {
               \node (p\i) at (\s * \x, 2 * \v) {\p} ;
               \node (s\i) at (\s * \x, \v) {\footnotesize{\sat}} ;
               \node (u\i) at (\s * \x, 0) {\footnotesize{\unsat}} ;
               \node (i\i) at (\s * \x, - .8 * \v) {\inp} ;
    }
    \node[draw, rounded corners, fit=(s1) (u1)] (nae) {} ;
    \foreach \i/\p/\q/\l/\m [count = \ip from 2] in {
      1/lem:sat-to-bo/lem:bo-to-sat/0.95/0.18,
      2/lem:balanced2sim-width/lem:simwidth2balanced/0.2/0.25,
      3/sec:lin-mim-balancing-to-lin-mim-width/sec:sim-width-to-sim-balancing/0.17/1.4}{
      \draw[arr] (s\i) to node[midway] {\hyperref[\p]{\rule[.5ex]{\l cm}{0pt}}} (s\ip) ;
      \draw[arr] (u\i) to node[midway] {\hyperref[\q]{\rule[.5ex]{\m cm}{0pt}}} (u\ip) ;
    }

  \end{tikzpicture}
  }
  \caption{Visual summary of our reduction, split into its three steps.}
  \label{fig:roadmap}
\end{figure}
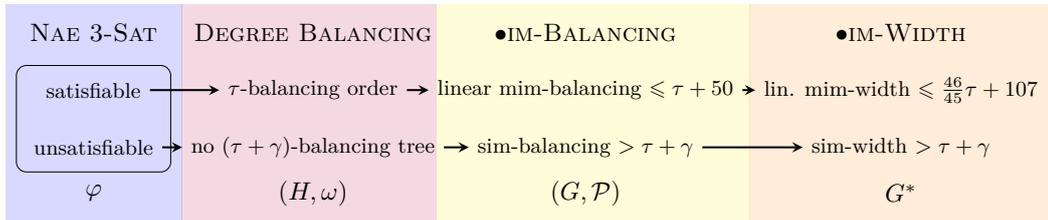

We now outline each step.

\medskip

\textbf{\textsc{Nae 3-Sat} to \textsc{Degree Balancing}.}
We actually reduce from the positive variant of \textsc{Nae 3-Sat}, where no literal is negated.
We design a~gadget called \emph{bottleneck sequence} that, given three disjoint sets $X, Y, Z \subset V(H)$, forces all vertices of $Y$ to appear in the order after all the vertices of~$X$, and before all the vertices of~$Z$ (or by symmetry after all the vertices of~$Z$, and before all the vertices of~$X$).
Vertices of $Y$ are in one-to-one correspondence with clauses of~$\varphi$.
Similarly, we have a~vertex for each variable of~$\varphi$.
Each \emph{variable} vertex is forced to be placed before $X$ (where it represents being set to true), or after $Z$ (where it represents being set to false).
The weights are designed so that a~\emph{clause} vertex can tolerate two but not three of its \emph{variable} vertices to be on the same side (before or after~it); which exactly captures the semantic of a~not-all-equal 3-clause.

The gap between \emph{at~most~$\tau$} and \emph{at~least~$\tau + \gamma +1$} is obtained by carefully crafting~$\weight$. 
We also add a~padding gadget to raise the minimum degree of~$H$, in such a~way that only two vertices have low-enough degree to be leaves of~$T$.
This forces $T$ to be a~path (the only tree with at~most~two leaves), thus \ldb and \tdb to coincide.

\medskip

\textbf{\textsc{Degree Balancing} to \textsc{\{Linear M, Tree S\}im-Balancing}.}
Every vertex $u$ of~$H$ becomes an independent set $S(u)$ of $G$ and a~part of $\mathcal P$ of size the sum of the weights of edges incident to~$u$.
Adjacencies in~$H$ become induced matchings in $G$, whereas non-adjacencies in~$H$ become bicliques in~$G$ (with some additional twist, see~\cref{fig:matching-and-dummy}).
The density of~$G$ forces large induced matchings to be mainly incident to a~single part~$S(u)$.
Thus, roughly speaking, the maximum induced matchings in $G$ behave like the degree in~$H$.
As the parts $S(u)$ are independent sets, there is in effect no difference between \tmb and \tsb.
The indifference between the tree or the linear variants is inherited from the previous reduction.
The actual arguments incur a~small additive loss (of~50) in the induced matching size, which is eventually outweighed by~$\gamma$.

\medskip

\textbf{\textsc{\{Linear M, Tree S\}im-Balancing} to \textsc{\{Linear M, S\}im-Width}.}
We design a~\emph{part gadget} $\mathcal G(u)$ that simultaneously takes care of the three items above~\cref{fig:roadmap}.
Essentially, every part $S(u)$ is transformed into the 1-subdivision $P_u$ of a~path on $|S(u)|$ vertices, then duplicated a~large (but constant) number of times, concatenated into a~single path, and every pair of vertices in different copies are linked by an edge whenever they do not correspond to the same vertex or neighboring vertices in $P_u$.
On the one hand, this may only increase the linear mim-width (compared to the linear mim-balancing) by an additive constant.
Following the ``spine'' of~$\mathcal G(u)$, one gets a~witness of low linear mim-width for $G^*$ from a~witness of low linear mim-balancing of~$(G,\mathcal P)$.

On the other hand, the dense ``path-like'' structure of $\mathcal G(u)$ ensures that, in an optimal decomposition of~$G^*$, its vertices may as well be placed in order at the leaves of a~caterpillar.
We thus devise a~process that builds a~witness of low sim-balancing for $(G,\mathcal P)$ from a~witness of low sim-width for~$G^*$:
We in turn identify an edge $e$ of the branch decomposition of~$G^*$ that can support $V(\mathcal G(u))$, and in particular $S(u)$, without increasing the width.
We then relocate the vertices of $S(u)$ at a~vertex subdividing~$e$.
Eventually each set $S(u)$ is solidified at a~single node of the tree, and we reach the desired witness for \tsb.

\medskip

\textbf{Remarks and perspectives.}
It can be noted that we had to develop completely new techniques.
Indeed, the known W[1]-hardness~\cite{SaetherV15,SaetherV16} relies on the difficulty of actually computing the value of a~fixed cut, i.e., solving \textsc{Maximum Induced Matching}.
In some sense, the instances produced there are not difficult to solve (a~best decomposition is, on the contrary, suggested by the reduction), but only to evaluate.
In any case, as \textsc{Maximum Induced Matching} is W[1]-hard but admits a~straightforward XP algorithm, we could not use the same idea.

We believe that \textsc{Linear Degree Balancing} could be explored for its own sake.
We emphasize that our techniques could prove useful to show the \paraNP-hardness of other parameters based on branch decompositions such as those recently introduced by Eiben et al.~\cite{Eiben22}, where the \emph{cut function} can combine maximum (semi-)induced matchings, maximum (semi-)induced co-matchings, maximum half-graphs (or ladders).
One would then mainly need to tune the gadgets of the second step (see~\cref{fig:matching-and-dummy}) to fit the particular cut function.

We notice that our reduction from \textsc{4-Occ Not-All-Equal 3-Sat} is linear: $n$-variables instances are mapped to $\Theta(n)$-vertex graphs.
Hence, unless the Exponential-Time Hypothesis (ETH)~\cite{Impagliazzo01} fails,\footnote{the assumption that there is a~$\lambda>0$ such that $n$-variable \textsc{3-Sat} cannot be solved in time $O(\lambda^n)$.} no $2^{o(n)}$-time algorithm can decide if the mim-width (or any of the five variants of mim-width) of an $n$-vertex graph is at~most 1211.
Indeed the absence of $2^{o(n)}$-time algorithm for $n$-variable \textsc{4-Occ Not-All-Equal 3-Sat} (even \textsc{Positive 4-Occ Not-All-Equal 3-Sat}) under the ETH can be derived from the Sparsification Lemma~\cite{sparsification} and classic reductions.

Our focus was to handle all the variants of mim-width at once.
This made the reduction more technical and degraded the constant upper and lower bounds.
Better bounds (than 1211) could be achieved if separately dealing with \textsc{Mim-Width} or with \textsc{Linear Mim-Width}.
For example, the latter problem essentially only requires the first two steps of the reduction.
Still, deciding if the (linear) mim-width of a~graph is at~most~1 (or any 1-digit constant) remains open.
In addition, the question whether an XP $f(\text{OPT})$-approximation algorithm for \textsc{Mim-Width} (and its variants) exists, for some fixed function~$f$ and $\text{OPT}$ being the optimum width, remains open. 

\section{Graph definitions and notation}

For $i$ and $j$ two integers, we denote by $[i,j]$ the set of integers that are at~least~$i$ and at~most~$j$.
For every integer $i$, $[i]$ is a shorthand for $[1,i]$.

\subsection{Standard graph theory}

We denote by $V(G)$ and $E(G)$ the vertex and the edge set, respectively, of a graph~$G$.
If $G$ is a~graph and $S \subseteq V(G)$, we denote by $G[S]$ the subgraph of $G$ induced by~$S$, and use $G-S$ as a~short-hand for $G[V(G) \setminus S]$.
If $e \in E(G)$, we denote by $G-e$ the graph $G$ deprived of edge $e$, but the endpoints of $e$ remain.
More generally, if $F \subseteq E(G)$, $G-F$ is the graph obtained from $G$ by removing all the edges of~$F$ (but not their endpoints).
For $X \subseteq V(G)$, we may denote by $E_G(X)$ the edge set of $G[X]$.
We denote the open and closed neighborhoods of a vertex $v$ in $G$ by $N_G(v)$ and $N_G[v]$, respectively.
For $S \subseteq V(G)$, we set $N_G(S) := \bigcup_{v \in S}N_G(v) \setminus S$ and $N_G[S] := N_G(S) \cup S$.
In every notation with a~graph subscript, we may omit it if the graph is clear from the context.
A~vertex set $S \subseteq V(G)$ \emph{covers} an edge set $F \subseteq E(G)$ if every edge of $F$ has at~least one endpoint in~$S$. 

A~\emph{cut} of a~graph $G$ is a~bipartition $(A,B)$ of~$V(G)$.
The \emph{cut-set} defined by a cut $(A,B)$, denoted by $E(A,B)$, is $\{uv\in E(G)\mid u\in A, v\in B\}$.
We denote by $G[A,B]$ the bipartite subgraph of $G$ with edge set $E(A,B)$.
A \emph{matching} is a set of edges that share no endpoints
and an \emph{induced matching} of $G$ is a matching $M$ such that every edge of $G$ intersects at~most one edge in $M$.
If $A, B \subseteq V(G)$ are two disjoint vertex subsets of $G$, a~\emph{matching between $A$ and $B$} is a~matching where every edge has one endpoint in~$A$ and the other endpoint in~$B$.
An induced matching in $G[A,B]$ is called a \emph{semi-induced} matching of $G$ between $A$ and $B$.

\subsection{Mim-width and its variants}

A~\emph{branch decomposition} or \emph{tree layout} (or simply \emph{layout}) of a~graph $G$ is a pair $(T, f)$ where $T$ is a~ternary tree (i.e., every internal node of~$T$ has degree~3) and $f$ is a~bijection from $V(G)$ to the leaves of~$T$.
Given two disjoint sets $X, Y \subseteq V(G)$, we denote by $\mim_G(X,Y)$ (resp.~$\ssim_G(X,Y)$) the maximum number of edges in a~semi-induced matching (resp.~induced matching) of $G$ between $X$ and $Y$, and may refer to it as \emph{mim-value} (resp.~\emph{sim-value}).
An edge $e$ of $T$ \emph{induces} or \emph{defines} a~cut $(A_e, B_e)$ of~$G$, where $A_e$ and $B_e$ are the preimages by $f$ of the leaves in the two components of $T-e$.

The \emph{mim-value} (resp.~\emph{sim-value}) of $(A_e, B_e)$ is set as $\mim_G(A_e,B_e)$ (resp.~$\ssim_G(A_e,B_e)$).
The \emph{mim-value} (resp.~\emph{sim-value}) of the branch decomposition $(T, f)$ is the maximum of $\mim_G(A_e,B_e)$ (resp.~$\ssim_G(A_e,B_e)$) taken over every edge $e$ of~$T$.
Finally, the \emph{mim-width} (resp.~\emph{sim-width}) of $G$ is the minimum \emph{mim-value} (resp.~\emph{sim-value}) taken over every branch decomposition $(T, f)$ of~$G$.

The \emph{upper-induced matching number} of $X \subseteq V(G)$ is the maximum size of an induced matching of $G - E(V(G) \setminus X)$ between $X$ and $V(G) \setminus X$.
The \emph{one-sided mim-width} is defined as above with the \emph{omim-value} of cut $(A_e, B_e)$, $\omim_G(A_e,B_e)$, defined as the minimum between the upper-induced matching numbers of~$A_e$ and of~$B_e$.

The linear variants of these widths and values impose $T$ to be a~rooted \emph{full} binary tree (i.e, every internal node has exactly two children) such that the internal nodes form a~path. 

\section{\textsc{Not-All-Equal 3-Sat} to \textsc{Degree Balancing}}

Given a~graph $H$ edge-weighted by a~map $\weight : E(H) \rightarrow \Nn$, the \emph{weight of a~vertex} $v$ of~$H$ is the sum of the weights of the edges incident to~$v$.
We say that a~total order $\prec$ on $V(H)$ is~\emph{$\tau$-balancing}, for some non-negative integer $\tau$,
if for every vertex $v \in V(H)$ the \emph{left weight of $v$}, $\sum_{u \in N(v), u \prec v} \weight(uv)$, and the \emph{right weight of $v$}, $\sum_{u \in N(v), v \prec u} \weight(uv)$, are both at most $\tau$, i.e., 
\[
\Delta_\prec(v) := \max \left( \sum_{u \in N(v), u \prec v} \weight(uv),  \sum_{u \in N(v), v \prec u} \weight(uv)\right) \le \tau.
\]

\medskip

\textbf{Constants $\bm{\tau}$, $\bm{\gamma}$, $\bm{\lambda}$.}
Henceforth we will use $\tau$ and $\gamma$ as global natural constants.
The reduction in this section will also use a~constant positive integer~$\lambda$.
For the current section, we need that the following conditions hold.
\begin{equation}\label{def:gamma-lambda}
 \gamma < \lambda,~~3\gamma + 4 < \tau,~~2\lambda + \gamma < \tau,~~6 \lambda \leqslant \tau.
\end{equation}

We will not only prove that \ldb is \paraNP-hard but we will obtain a~scalable additive gap. 
More specifically, we start by showing the following.

\begin{theorem}\label{thm:ldb}
  It is \NP-hard to distinguish graphs having \mbox{a~$\tau$-balancing} order from graphs having no $(\tau + \gamma)$-balancing order.
\end{theorem}

Eventually we will need that $\tau$ and $\gamma$ are multiples of a~constant integer $a$ (which is defined and set to~45 in~\cref{sec:balancing-to-mim-width}).
This can simply be achieved by multiplying all edge weights of the forthcoming reduction by~$a$.
We will finally set $\tau := 24a = 1080$, $\lambda := 4a = 180$, and $\gamma := 3a = 135$.
One can quickly check that these values do respect~\cref{def:gamma-lambda}.

\subsection{First properties on balancing orders, and bottlenecks}

Given a total order $\prec$ on a graph $H$, we say that a~vertex set $S$ is \emph{smaller} (resp. \emph{larger}) than another vertex set $U$, denoted by $S \prec U$ (resp.~$U \prec S$), if for all $s \in S, u \in U$ we have $s \prec u$ (resp. $u \prec s)$.
When a~set $S$ is neither larger nor smaller than a vertex $u$, we say that $S$ \emph{surrounds}~$u$.
We also say that $u$ \emph{is surrounded by}~$S$.
Note that if $S$ surrounds two vertices $u$ and $v$, it surrounds any vertex $w$ with $u \prec w \prec v$.

We begin with a~useful observation on the only possible $\tau$-balancing order of a~$P_3$ (i.e., 3-vertex path) with large total weight.

\begin{lemma}\label{lem:P3-order}
For any integer $t$, and any edge-weighted graph $(H,\weight)$ containing a~$P_3$ $abc$ such that $\weight(ab) + \weight(bc) > t$.
Then in any $t$-balancing order of~$H$, $\{a, c\}$ surrounds~$b$.
\end{lemma}
\begin{proof}
If $\{a, c\} \prec b$ is (resp. $b \prec \{a, c\}$), then the left (resp. right) weight of~$b$ is more than~$t$.
\end{proof}

We also rely on the following observation, where the induced subgraph of an edge-weighted graph $(H,\weight)$ is an induced subgraph of $H$ edge-weighted by the restriction of $\weight$ to its edge set.
\begin{observation}\label{obs:bo-ind-sub}
  Every $t$-balancing order of $(H,\weight)$ is a~$t$-balancing order of any induced subgraph of~$(H,\weight)$. 
\end{observation}

Our main ingredient here is called~\emph{bottleneck}.

\begin{definition}\label{def:bottleneck}
  A~\emph{$(\tau, \gamma)$-bottleneck on terminals $v_1, \ldots, v_k$} is an~edge-weighted caterpillar $B$ defined as follows.
\begin{compactenum}
\item Let $P(B)$ be a~$2k$-vertex path, say $a_1b_1a_2b_2 \dots a_kb_k$, called \emph{spine} of~$B$ and for every $i \in [k]$, we set $\weight(a_ib_i) := \tau$ and $\weight(b_ia_{i+1}) := \gamma + 1$.
\item We obtain $B$ by adding to $P(B)$ a~leaf $v_i$ adjacent to $a_i$, satisfying $\gamma + 1 \leqslant \weight(v_ia_i) \leqslant \tau - \gamma - 1$ for every $i \in [k]$.
  This edge is called the \emph{attachment of $v_i$ to $B$}.
  \item The caterpillar $B$ is rooted in~$b_k$.
\end{compactenum}
\end{definition}

The vertex $v_1$ is called \emph{first terminal of $B$}.
A~$(\tau, \gamma)$-bottleneck is depicted in~\cref{fig:bottleneck}.

\begin{figure}[h!]
\centering
  \begin{tikzpicture}[scale=.9, vertex/.style={draw, circle, minimum size=6mm, inner sep=0pt}, every path/.style={thick}]
    \def\s{1.6}

\node[vertex] (a1) at (0, 0) {$a_1$};
\node[vertex] (b1) [right=\s cm of a1] {$b_1$};
\node[vertex] (a2) [right=\s cm of b1] {$a_2$};
\node[vertex] (b2) [right=\s cm of a2] {$b_2$};
\node (dots) [right=\s cm of b2] {$\dots$};
\node[vertex] (a3) [right=\s cm of dots] {$a_k$};
\node[vertex] (b3) [right=\s cm of a3] {$b_k$};

\node[rectangle, draw=none, minimum size=0mm, inner sep=0pt] (root) [below=0.5cm of b3] {};

\draw (a1) -- node[above] {$\tau$} (b1);
\draw (b1) -- node[above] {$\gamma + 1$} (a2);
\draw (a2) -- node[above] {$\tau$} (b2);
\draw (b2) -- node[above] {$\gamma + 1$} (dots);
\draw (dots) -- node[above] {$\gamma + 1$} (a3);
\draw (a3) -- node[above] {$\tau$} (b3);

\foreach \i in {1}{
  \node[vertex] (v\i) [below = 1cm of a\i] {$v_\i$};
  \draw (v\i) -- node[right, midway] {$\weight(a_\i v_\i) \in [\gamma+1,\tau-\gamma-1]$} (a\i);
}
\foreach \i/\j in {2/2,3/k}{
  \node[vertex] (v\i) [below = 1cm of a\i] {$v_\j$};
  \draw (v\i) -- node[right, midway] {$\weight(a_\j v_\j)$} (a\i);
}
\node [above = 0.15cm of b3] {root};

\end{tikzpicture}
\caption{Illustration of a~$(\tau,\gamma)$-bottleneck.}
\label{fig:bottleneck}
\end{figure}
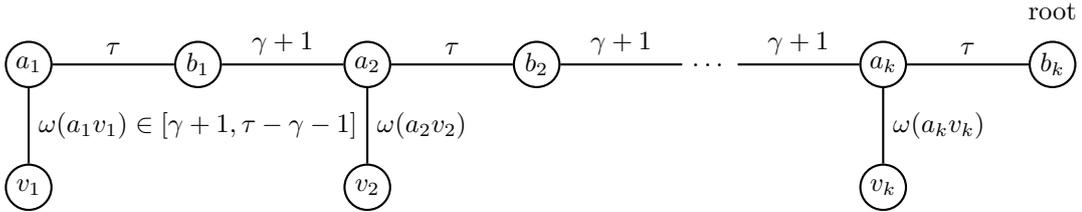

A~bottleneck ensures the following.

\begin{lemma}\label{lem:bottle-order}
Let $\prec$ be a~$(\tau + \gamma)$-balancing order of a~$(\tau, \gamma)$-bottleneck $B$ on terminals $v_1, \dots, v_k$.
Using the notation of~\cref{def:bottleneck}, if $a_k \prec b_k$ then $a_1 \prec b_1 \prec a_2 \prec b_2 \prec \dots \prec a_k \prec b_k$, and $v_i \prec a_i$ for each $i \in [k]$.
Hence symmetrically, if $b_k \prec a_k$ then $b_k \prec a_k \prec b_{k-1} \prec a_{k-1} \prec \dots \prec b_1 \prec a_1$, and $a_i \prec v_i$ for each $i \in [k]$.
\end{lemma}

\begin{proof}
  Let $i$ be the smallest index such that $a_i \prec b_i \prec a_{i+1} \prec b_{i+1} \prec \dots \prec a_k \prec b_k$.
  Assume for the sake of contradiction that $i \geqslant 2$.
  Since $\weight(a_ib_i) = \tau$ and  $\min(\weight(a_ib_{i-1}), \weight(a_iv_i)) \geqslant \gamma+1$, it holds that $b_{i-1} \prec a_i$ and $v_i \prec a_i$, by applying~\cref{lem:P3-order} on $b_{i-1} a_i b_i$ and $v_i a_i b_i$.

We have $\weight(a_{i-1}b_{i-1}) + \weight(b_{i-1}a_i) = \tau + \gamma + 1$, so, by~\Cref{lem:P3-order} on the $P_3$ $a_{i-1}b_{i-1}a_i$, $a_{i-1} \prec b_{i-1} \prec a_i$.
This contradicts the minimality of~$i$; thus we have $i = 1$.
And in particular, we also have $v_i \prec a_i$ for every $i \in [k]$.
\end{proof}

Henceforth every bottleneck is a~$(\tau, \gamma)$-bottleneck.
Thus we simply write \emph{bottleneck}.

\begin{definition}\label{def:bottleneck-seq}
  Given three vertex sets $S_1, S_2, S_3$, we call \emph{bottleneck sequence on $S_1, S_2, S_3$} an edge-weighted graph $B(S_1, S_2, S_3)$ obtained by adding
  \begin{compactenum}
  \item for every $i \in \{1,2\}$, a~bottleneck $B_i^+$ with terminals $S_i \cup \{s_i\}$ where $s_i$ is the first terminal of $B_i^+$, and the attachment of $s_i$ is of weight $\gamma + 1$, 
  \item for every $i \in \{2,3\}$, a~bottleneck $B_i^-$ with terminals $S_i \cup \{s_i\}$ where $s_i$ is the first terminal of $B_i^-$ and the attachment of $s_i$ is of weight $\gamma + 1$ such that
  \item for every $i \in \{1,2\}$, the roots of $B_i^+$ and of $B_{i+1}^-$ are identified as the same vertex, and
  \item for every $i \in \{2,3\}$, an edge $s_is_{i+1}$ of weight $\lfloor \frac{\tau + \gamma}{2} \rfloor + 1$, 
  \end{compactenum}
  with $s_1, s_2, s_3$ three new vertices.
\end{definition}

\begin{figure}[h!]
  \centering
  \begin{tikzpicture}[scale=.9, vertex/.style={draw, circle, minimum size=4.5mm, inner sep=0pt}, every path/.style={thick}]
	\def\s{1.1}
	\def\fo{0.08}
	
	\begin{scope}   
		\node[vertex] (a1) at (0, 0) {};
		\node[vertex] (b1) [right=\s cm of a1] {};
		\node[vertex] (a2) [right=\s cm of b1] {};
		\node[vertex] (b2) [right=\s cm of a2] {};
		\node (dots) [right=\s cm of b2] {$\ldots$};
		\node[vertex] (a3) [right=\s cm of dots] {};
		\node[circle, minimum size=4.5mm, inner sep=0pt] (bi3) [right=\s cm of a3] {};
		\node[vertex] (b3) at ([xshift=1.78 cm, yshift=-2.06 cm] a3) {};
		
		\draw (a1) -- node[below] {$\tau$} (b1);
		\draw (b1) -- node[below] {$\gamma + 1$} (a2);
		\draw (a2) -- node[below] {$\tau$} (b2);
		\draw (b2) -- node[below] {$\gamma + 1$} (dots);
		\draw (dots) -- node[below] {$\gamma + 1$} (a3);
		\draw (a3) -- node[left] {$\tau$} (b3);
		
		\foreach \i in {1}{
			\node[vertex, preaction={fill opacity = 0.5, fill = blue}] (v\i) [below = 1cm of a\i] {$s_\i$};
			\draw (v\i) -- node[right, midway] {$\gamma+1$} (a\i);
		}
		\foreach \i/\j in {2/2,3/k}{
			\node[vertex, fill opacity = 0.5, fill = blue] (v\i) [below = 1cm of a\i] {};
			\draw (v\i) -- node[right, midway] {} (a\i);
		}
		\node [right = 0.05cm of b3] {root};
		\node[draw,thick,rounded corners,fit=(a1) (bi3) (v3), fill opacity = \fo, fill=blue] (B1p) {} ;
		\node[fit=(a1) (bi3) (v3)] (B1+) {} ;
		\node at ([yshift=-0.25 cm] B1+) {$B^+_1$};
	\end{scope}
	
	\begin{scope}[yshift=-2.52cm]
		\node[vertex] (a1) at (0, 0) {};
		\node[vertex] (b1) [right=\s cm of a1] {};
		\node[vertex] (a2) [right=\s cm of b1] {};
		\node[vertex] (b2) [right=\s cm of a2] {};
		\node (dots) [right=\s cm of b2] {$\ldots$};
		\node[vertex] (a3) [right=\s cm of dots] {};
		\node[circle, minimum size=4.5mm] (bi3) [right=\s cm of a3] {};
		
		\node[rectangle, draw=none, minimum size=0mm, inner sep=0pt] (root) [below=0.5cm of b3] {};
		
		\draw (a1) -- node[below] {$\tau$} (b1);
		\draw (b1) -- node[below] {$\gamma + 1$} (a2);
		\draw (a2) -- node[below] {$\tau$} (b2);
		\draw (b2) -- node[below] {$\gamma + 1$} (dots);
		\draw (dots) -- node[below] {$\gamma + 1$} (a3);
		
		\foreach \i in {1}{
			\node[vertex, preaction={fill opacity = 0.5, fill = green}] (vv\i) [below = 1cm of a\i] {$s_2$};
			\draw (vv\i) -- node[right, midway] {$\gamma+1$} (a\i);
		}
		\foreach \i/\j in {2/2,3/k}{
			\node[vertex, fill opacity = 0.5, fill = green] (vv\i) [below = 1cm of a\i] {};
			\draw (vv\i) -- node[right, midway] {} (a\i);
		}
		\draw (a3) -- node[below] {$\tau$} (b3);
		\node[draw,thick,rounded corners,fit=(a1) (bi3) (vv3), fill opacity = \fo, fill = green] (B2m) {} ;
		\node[fit=(a1) (bi3) (vv3)] (B2+) {} ;
		\node at ([yshift=-0.1 cm] B2+) {$B^-_2$};
		
		\foreach \i/\j in {1/$\gamma+1$,2/,3/}{
			\node[vertex] (w\i) [below=1.1 cm of vv\i] {} ;
			\draw (vv\i)  -- node[right, midway] {\j} (w\i) ;
		}
		
		\foreach \i in {1,2}{
			\node[vertex] (x\i) [right=\s cm of w\i] {} ;
			\draw (w\i)  -- node[above, midway] {$\tau$} (x\i) ;
		}
		\draw (x1)  -- node[above, midway] {$\gamma+1$} (w2) ;
		\node[circle, minimum size=4.5mm] (wd) [right=\s cm of x2] {$\ldots$} ;
		\draw (x2) -- node[above, midway] {$\gamma+1$} (wd) ;
		\draw (wd) -- node[above, midway] {$\gamma+1$} (w3) ;
		\node[circle, minimum size=4.5mm] (w4) [right=\s cm of w3] {} ;
		
		\node[draw,thick,rounded corners,fit=(vv1) (w4), fill opacity = \fo, fill = green] (B2p) {} ;
		\node[draw,thick,rounded corners,fit=(vv1) (w4)] {$B^+_2$} ;
		
		\foreach \i in {1}{
			\node[vertex, preaction={fill opacity = 0.5, fill = red}] (y\i) [below=0.3 cm of w\i] {$s_3$} ;
			\node[vertex] (z\i) [below=1.1 cm of y\i] {} ;
		}
		\foreach \i in {2,3}{
			\node[vertex] (y\i) [below=0.3 cm of w\i, fill opacity = 0.5, fill = red] {} ;
			\node[vertex] (z\i) [below=1.1 cm of y\i] {} ;
		}
		\foreach \i in {1,2}{
			\node[vertex] (zp\i) [right=\s cm of z\i] {} ;
			\draw (z\i)  -- node[above, midway] {$\tau$} (zp\i) ;
		}
		
		\draw (y1)  -- node[right, midway] {$\gamma+1$} (z1) ;
		\foreach \i in {2,3}{
			\draw (y\i) -- (z\i) ;
		}
		
		\node[circle, minimum size=4.5mm] (zpp) [right=\s cm of zp2] {$\ldots$} ;
		\draw (zp1) -- node[above, midway] {$\gamma+1$} (z2) ;
		\draw (zp2) -- node[above, midway] {$\gamma+1$} (zpp) ;
		\draw (zpp) -- node[above, midway] {$\gamma+1$} (z3) ;
		\node[circle, minimum size=4.5mm] (zf) [right=\s cm of z3] {} ;
		\node[circle, minimum size=4.5mm] (ze) [right=1.13 cm of z3] {} ;
		
		\node[draw,thick,rounded corners,fit=(y1) (zf), fill opacity = \fo, fill = red] (B3m) {} ;
		\node[draw,thick,rounded corners,fit=(y1) (zf)] {$B^-_3$} ;
		
		\node[vertex] (root2) [above=1.5 cm of ze] {} ;
		
		\draw (z3)  -- node[right, midway] {$\tau$} (root2) ;
		\draw (w3)  -- node[above, midway] {$\tau$} (root2) ;
		\node [right = 0.05cm of root2] {root};
	\end{scope}
	
	\foreach \i/\j in {v1/vv1, vv1/y1}{
		\draw (\i) to [bend left = -40] node[midway,left] {$\lfloor \frac{\tau+\gamma}{2} \rfloor + 1$} (\j) ;
	}
	
\end{tikzpicture}
  \caption{Bottleneck sequence $B(S_1,S_2,S_3)$.
    Vertices of $S_1 \cup \{s_1\}, S_2 \cup \{s_2\}, S_3 \cup \{s_3\}$ are in red, green, and blue, respectively.
    As in~\cref{fig:bottleneck}, every edge with an unspecified weight get one in the discrete interval $[\gamma+1,\tau-\gamma-1]$.}
\label{fig:bottleneck-sequence}
\end{figure}
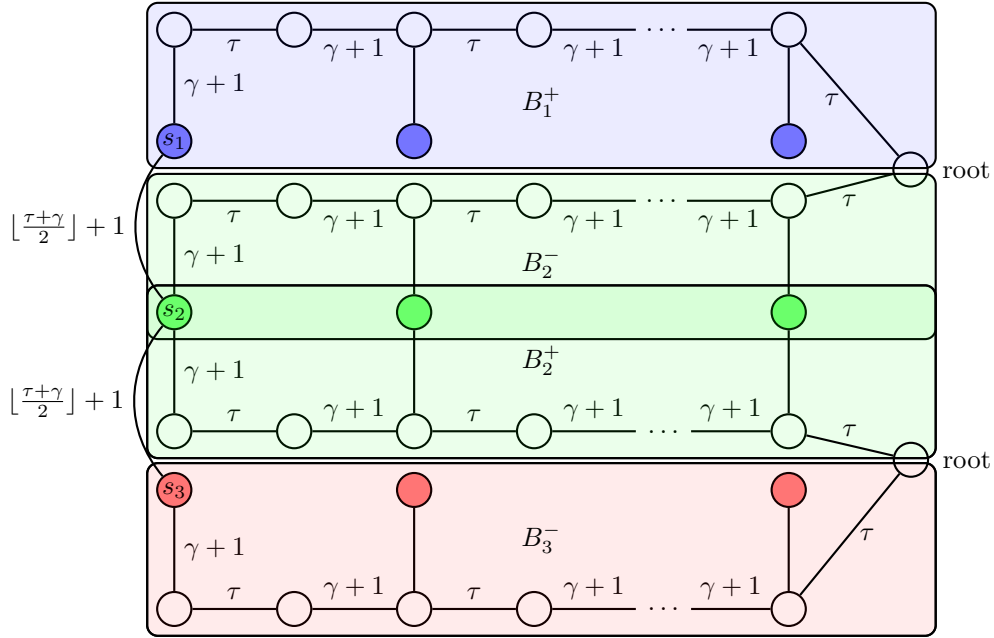
The next lemma yields the crucial property ensured by bottleneck sequences.

\begin{lemma}\label{lem:Forcing-the-order}
Any $(\tau + \gamma)$-balancing order $\prec$ on a~bottleneck sequence $B(S_1, S_2, S_3)$ is such that $S_1 \prec S_2 \prec S_3$ or $S_3 \prec S_2 \prec S_1$.
\end{lemma}
\begin{proof}
  We keep the notation of~\cref{def:bottleneck-seq}.
Applying~\cref{lem:P3-order} on the $P_3$ $s_1s_2s_3$, we get that either $s_1 \prec s_2 \prec s_3$ or $s_3 \prec s_2 \prec s_1$.

For $i \in \{1,2\}$, let $r_i$ be the common root of $B_i^+$ and $B_{i+1}^-$.
Vertex $r_i$ has exactly two neighbors: a~vertex $a^+ \in V(B_i^+)$ and a~vertex $a^- \in V(B_{i+1}^-)$.
By construction $\weight(a^-r_i) = \weight(a^+r_i)= \tau$, and by \Cref{lem:P3-order} (since $2 \tau > \tau + \gamma$), either $a^- \prec r_i \prec a^+$ or $a^+ \prec r_i \prec a^-$.
By \Cref{lem:bottle-order}, this implies that $S_{i+1} \cup \{s_{i+1}\} \prec S_i \cup \{s_i\}$ or $S_i \cup \{s_i\} \prec S_{i+1} \cup \{s_{i+1}\}$.

In particular, $S_1 \prec S_2 \prec S_3$ (if $s_1 \prec s_2 \prec s_3$) or  $S_3 \prec S_2 \prec S_1$ (if $s_3 \prec s_2 \prec s_1$).
\end{proof}

We conclude the section by defining $\tau$-balancing orders for bottleneck sequences.
A~\emph{direct order} $\prec_{\rightarrow}$ of a~bottleneck $B$ with terminals $v_1, \dots, v_k$ goes as follows:
$$
\{v_1, \dots, v_k\} \prec_{\rightarrow} a_1 \prec_{\rightarrow} b_1 \prec_{\rightarrow} a_2 \prec_{\rightarrow} b_2 \prec_{\rightarrow} \dots \prec_{\rightarrow} a_k \prec_{\rightarrow} b_k,
$$
where $a_1b_2 \dots a_kb_k$ is the spine of $B$ rooted in~$b_k$.
Note that the order induced by $\{v_1, \dots, v_k\}$ is not specified (and so a~given bottleneck on $k$ terminals admits $k!$ different direct orders).
A~\emph{reverse order of~$B$}, denoted by $\prec_{\leftarrow}$, is simply defined as the reverse order of a~direct order~$\prec_{\rightarrow}$.

A~\emph{direct order} $\prec^\text{seq}_{\rightarrow}$ of a~bottleneck sequence $B(S_1, S_2, S_3)$ is a~common (linear) extension of direct orders on $B_1^+$ and $B_2^+$ and reverse orders on $B_2^-$ and $B_3^-$.
Note that on any bottleneck sequence, at~least one direct order exists since the direct and reverse orders constrain disjoint vertex sets. 
In particular we have $S_1 \prec^\text{seq}_{\rightarrow} S_2 \prec^\text{seq}_{\rightarrow} S_3$.
We check that any direct order of $B(S_1, S_2, S_3)$ is indeed $\tau$-balancing.

\begin{lemma}\label{lem:direct-order}
A direct order $\prec^\text{seq}_{\rightarrow}$ of the bottleneck sequence $B(S_1, S_2, S_3)$ is $\tau$-balancing.
\end{lemma}
\begin{proof}
Again we use the notation of~\cref{def:bottleneck-seq}.
For each $i \in [3]$, and any vertex $v \in S_i$, $v$~has at most two neighbors: a~vertex $t^- \in V(B_i^-)$ and a~vertex $t^+ \in V(B_i^+)$.
By construction, $t^- \prec^\text{seq}_{\rightarrow} v \prec^\text{seq}_{\rightarrow} t^+$ and both $\weight(t^-v)$ and $\weight(vt^+)$ are at most $\tau - \gamma - 1 \leqslant \tau$. 
For each $i\in [3]$, the vertex $s_i$ has at most four neighbors: $s_{i-1}$, $s_{i+1}$, a vertex $t^- \in V(B_i^-)$, and a vertex $t^+ \in V(B_i^+)$.
By construction, the left weight of $s_i$ is at most \[\weight(s_{i-1} s_i) + \weight(t^-s_i) = \left(\left\lfloor \frac{\tau + \gamma}{2} \right\rfloor + 1 \right) + (\gamma + 1) = \frac{\tau}{2} + \frac{3\gamma}{2} + 2 \le \tau.\]
The last inequality holds since $3\gamma \le \tau - 4$.
The right weight of $s_i$ can be symmetrically upper bounded by~$\tau$.

It remains to check the degree property for the vertices in bottleneck spines.
For each $i \in \{1,2\}$, let $a_1b_1 \ldots a_kb_k$ be the spine $P(B_i^+)$.
For any $j \in [k]$, vertex $a_j$ has at~most three neighbors: $b_{j-1}$ (if it exists), $b_j$, and some leaf $\ell \in S_i \cup \{s_i\}$. 
By construction, $\{\ell, b_{j-1}\} \prec^\text{seq}_{\rightarrow} a_j \prec^\text{seq}_{\rightarrow} b_j$.
Hence $a_j$ has left weight at most $(\tau - \gamma - 1) + (\gamma + 1) = \tau$, and right weight~$\tau$.
Vertex $b_j$ is incident to at most two edges each of weight at~most $\tau$ (even $b_k$).
By construction, one neighbor of $b_j$ is smaller and its other neighbor is larger.
Hence its left and right weights are both upper bounded by~$\tau$.

The case of vertices of $B_i^-$ with $i \in \{2,3\}$ is handled symmetrically.
\end{proof}

\subsection{Encoding \textsc{NAE 3-Sat} in \ldb}

We now describe the reduction from \textsc{Nae 3-Sat} to \ldb.
We recall that a~not-all-equal 3-clause is satisfied if it has at~least one satisfied literal and at~least one unsatisfied literal.
The \textsc{Nae 3-Sat} remains \NP-hard if each clause is on exactly three distinct \emph{positive} literals, and every variable appears exactly four times positively (and zero times negatively)~\cite{Darmann20}.
Let $\varphi$ be any such $n$-variable \textsc{Nae 3-Sat} instance.
As we will only deal with not-all-equal 3-clauses, we say that $\varphi$ is \emph{satisfiable} whenever it admits a~truth assignment that, in each clause of~$\varphi$, sets a~(positive) literal to true and another (positive) literal to false. 
We will build an edge-weighted graph $H := H(\varphi)$ as follows.


\medskip

\textbf{Variables, clauses, and variable-clause incidence.}
For each variable $x$ of~$\varphi$, we add a~vertex $v_x$ (to $H(\varphi)$), the \emph{vertex of $x$}.
For each clause $c$ of~$\varphi$ we add a vertex $v_c$, the \emph{vertex of $c$}.
For every clause $c$ and every variable $x$ in $c$, we add the edge $v_xv_c$ of~weight~$\lambda$.
We add two sets of vertices $T = \{t_i \ : \ i \in [n]\}$ and $F = \{f_i \ : \ i \in [n]\}$, for \emph{true} and \emph{false}.
For each~$t_i$ (resp. each~$f_i$), we add a vertex $\overline{t_i}$ (resp. a vertex $\overline{f_i}$), and the edge $t_i \overline{t_i}$ (resp. $f_i \overline{f_i}$) of~weight $\tau - \lambda$.
For each $i \in [n]$, let $x_i$ be the $i$-th variable of $\varphi$.
We add a~vertex $\overline{v_{x_i}}$, and the edges $v_{x_i}t_i, v_{x_i}f_i, \overline{v_{x_i}}t_i, \overline{v_{x_i}}f_i$ each of~weight~$\lambda$.

\medskip

\textbf{Bottleneck sequence $\bm{B(T,C,F)}$.}
We then add a~bottleneck sequence $B(T,C,F)$ where $C := \{v_c \ : \ c\text{ is a clause of } \varphi \}$, with weight $\tau - \lambda$ on every attachment incident to $T$ or $F$, and weight $\tau - 2\lambda$ on every attachment incident to~$C$.
(This is allowed since $\gamma + 1 \leqslant \tau - 2\lambda \leqslant \tau - \lambda \leqslant \tau - \gamma - 1$.)
We remind the reader that every attachment of the first terminals of the bottlenecks forming $B(T,C,F)$ has weight $\gamma+1$.
These three first terminals are extra vertices not in $T$, $C$, and~$F$.

\medskip

This could end the construction of $H$, but we want to impose an extra condition, which will later prove useful.
Specifically, we want that all \emph{but two} vertices have weight at~least~$\tau + \gamma + 1$ (both having weight~$\tau$).
Let us call $H'$ the edge-weighted graph built so far. 

\medskip

\textbf{Weight padding.}
For each vertex $v \in V(H')$ of weight less than $\tau + \gamma + 1$, the \emph{missing weight} of $v$ is defined as $\tau + \gamma + 1$ minus the weight of~$v$.
Let $p$ be the sum of missing weights of vertices of~$H'$.
Let $X$ and $Y$ be two sets each comprising~$p$ new vertices.
We add a~bottleneck $B_L$ with terminals the vertices of~$X$, and a~bottleneck $B_R$ with terminals the vertices of~$Y$.
Every attachment to $B_L$ and $B_R$ with an unspecified weight gets weight $\tau - \gamma - 1$.
We add a~perfect matching between $X$ and $Y$ with every edge of~weight $2 \gamma + 2$.
Finally, for each vertex $v \in V(H')$, we link $v$ by edges of weight 1 to $t$ vertices of $X$, where $t$ is the missing weight of $v$.
We do so such that every vertex in $X$ has exactly one neighbor in~$V(H')$.

This completes the construction of $H$; see~\cref{fig:first-red}.
We observe that $H$ is triangle-free.
This fact will significantly simplify some proof in~\cref{sec:balancing-to-mim-width} (although is not in any way crucial).

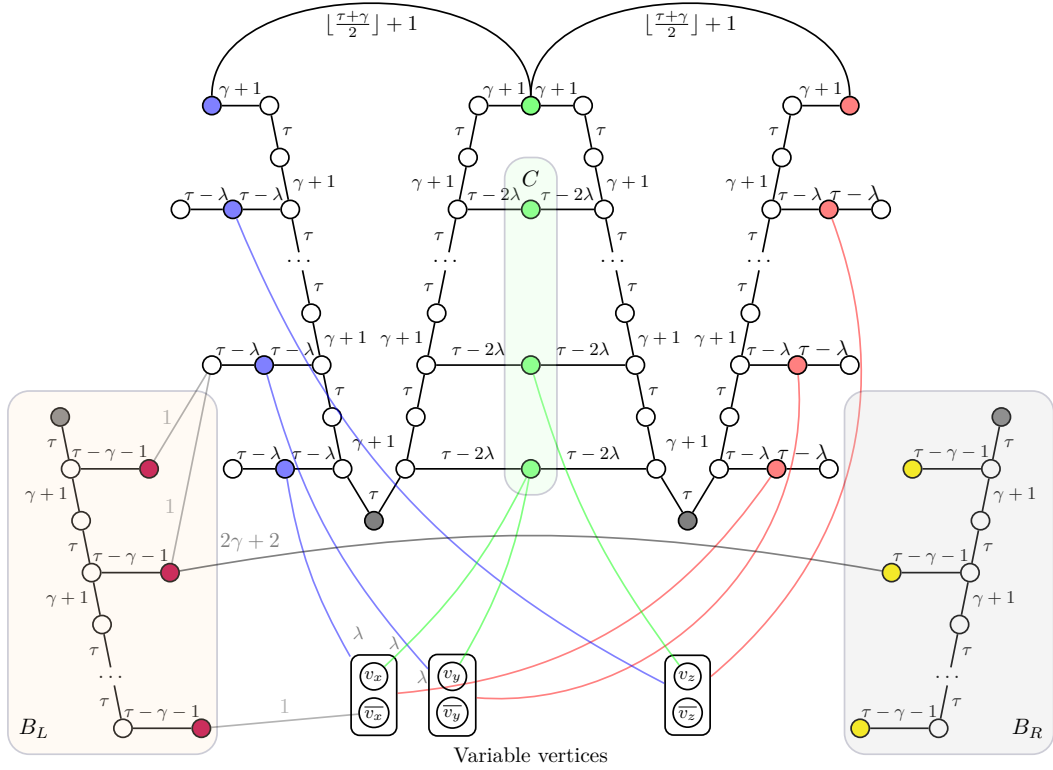
\begin{figure}[h!]
  \centering
  \scalebox{.8}{
  \begin{tikzpicture}[scale=.86, vertex/.style={draw, circle, minimum size=3mm, inner sep=0pt}, every path/.style={thick}]
\def\r{0.2}
\def\Lr{1.5}
\def\b{1}

\def\rootBLx{-1}
\def\rootBLy{0}

\begin{scope}[xshift=-0.5cm]

\node[vertex, fill = black!50!] (rootBL) at (\rootBLx, \rootBLy) {};

\foreach \i [count=\t from 1] in {1, 3, 6} {
	\node[vertex] (aBL\t) at (\i * \r + \rootBLx, -\i * \b + \rootBLy) {};
	\node[vertex, fill = purple] (sL\t) at (\i * \r + \Lr + \rootBLx, -\i * \b + \rootBLy) {};
}
\foreach \i [count=\t from 1] in {2, 4} {
	\node[vertex] (bBL\t) at (\i * \r + \rootBLx, -\i * \b + \rootBLy) {};
}
\node (dotsBL) at (5 * \r + \rootBLx, -5 * \b + \rootBLy) {$\dots$};

\end{scope}

\draw (rootBL) -- node[left] {\footnotesize $\tau$} (aBL1);
\draw (aBL1) -- node[left] {\footnotesize $\gamma + 1$} (bBL1);
\draw (bBL1) -- node[left] {\footnotesize $\tau$} (aBL2);
\draw (aBL2) -- node[left] {\footnotesize $\gamma + 1$} (bBL2);
\draw (bBL2) -- node[left] {\footnotesize $\tau$} (dotsBL);
\draw (dotsBL) -- node[left] {\footnotesize $\tau$} (aBL3);

\foreach \i in {1, 2, 3}{
  \draw (sL\i) -- node[above = 0, midway] {\footnotesize $\tau-\gamma-1$} (aBL\i);
}

\def\rootBRx{16}
\def\rootBRy{0}

\begin{scope}[xshift=0.5cm]

\node[vertex, fill = black!50!] (rootBR) at (\rootBRx, \rootBRy) {};

\foreach \i [count=\t from 1] in {1, 3, 6} {
	\node[vertex] (aBR\t) at (-\i * \r + \rootBRx, -\i * \b + \rootBRy) {};
	\node[vertex, fill = yellow] (sR\t) at (-\i * \r - \Lr + \rootBRx, -\i * \b + \rootBRy) {};
}
\foreach \i [count=\t from 1] in {2, 4} {
	\node[vertex] (bBR\t) at (-\i * \r + \rootBRx, -\i * \b + \rootBRy) {};
}
\node (dotsBR) at (-5 * \r + \rootBRx, -5 * \b + \rootBRy) {$\dots$};

\end{scope}

\draw (rootBR) -- node[right] {\footnotesize $\tau$} (aBR1);
\draw (aBR1) -- node[right] {\footnotesize $\gamma + 1$} (bBR1);
\draw (bBR1) -- node[right] {\footnotesize $\tau$} (aBR2);
\draw (aBR2) -- node[right] {\footnotesize $\gamma + 1$} (bBR2);
\draw (bBR2) -- node[right] {\footnotesize $\tau$} (dotsBR);
\draw (dotsBR) -- node[right] {\footnotesize $\tau$} (aBR3);

\foreach \i in {1, 2, 3}{
  \draw (sR\i) -- node[above = 0, midway] {\footnotesize $\tau-\gamma-1$} (aBR\i);
}

\def\rootTx{7.5 - 3}
\def\rootTy{-2}

\node[vertex, fill = black!50!] (rootT) at (\rootTx, \rootTy) {};

\foreach \i [count=\t from 1] in {1, 3, 6, 8} {
	\node[vertex] (aT\t) at (-\i * \r + \rootTx - 2 * \r, \i * \b + \rootTy) {};
	\node[vertex, fill opacity = 0.5, fill = blue] (sT\t) at (-\i * \r - \Lr + \rootTx, \i * \b + \rootTy) {};
	\node[vertex] (aCL\t) at (\i * \r + \rootTx + 2 * \r, \i * \b + \rootTy) {};
	\node[vertex, fill opacity = 0.5, fill = green] (sCL\t) at (7.5, \i * \b + \rootTy) {};
}
\foreach \i [count=\t from 1] in {1, 3, 6} {
  \node[vertex] (osT\t) at (-\i * \r - \Lr + \rootTx - 1, \i * \b + \rootTy) {};
  \draw (osT\t) to node[midway, above] {\footnotesize{$\tau-\lambda$}} (sT\t) ;
}
\foreach \i [count=\t from 1] in {2, 4, 7} {
	\node[vertex] (bT\t) at (-\i * \r + \rootTx - 2 * \r, \i * \b + \rootTy) {};
	\node[vertex] (bCL\t) at (\i * \r + \rootTx + 2 * \r, \i * \b + \rootTy) {};
}
\node (dotsT) at (-5 * \r + \rootTx - 2 * \r, 5 * \b + \rootTy) {$\dots$};
\node (dotsCL) at (5 * \r + \rootTx + 2 * \r, 5 * \b + \rootTy) {$\dots$};

\draw (rootT) -- node[right] {} (aT1);
\draw (aT1) -- node[right] {} (bT1);
\draw (bT1) -- node[right] {\footnotesize $\tau$} (aT2);
\draw (aT2) -- node[right] {\footnotesize $\gamma + 1$} (bT2);
\draw (bT2) -- node[right] {\footnotesize $\tau$} (dotsT);
\draw (dotsT) -- node[right] {\footnotesize $\tau$} (aT3);
\draw (aT3) -- node[right] {\footnotesize $\gamma + 1$} (bT3);
\draw (bT3) -- node[right] {\footnotesize $\tau$} (aT4);

\draw (rootT) -- node[left] {} (aCL1);
\draw (aCL1) -- node[left] {} (bCL1);
\draw (bCL1) -- node[left] {\footnotesize $\tau$} (aCL2);
\draw (aCL2) -- node[left] {\footnotesize $\gamma + 1$} (bCL2);
\draw (bCL2) -- node[left] {\footnotesize $\tau$} (dotsCL);
\draw (dotsCL) -- node[left] {\footnotesize $\tau$} (aCL3);
\draw (aCL3) -- node[left] {\footnotesize $\gamma+1$} (bCL3);
\draw (bCL3) -- node[left] {\footnotesize $\tau$} (aCL4);

\draw (sT4) -- node[above = 0, midway] {\footnotesize $\gamma+1$} (aT4);
\draw (sCL4) -- node[above = 0, midway] {\footnotesize $\gamma+1$} (aCL4);
\foreach \i in {1,2,3}{
	\draw (sT\i) -- node[above = 0, midway] {\footnotesize $\tau -\lambda$} (aT\i);
	\draw (sCL\i) -- node[above = 0, midway] {\footnotesize $\tau - 2\lambda$} (aCL\i);
}


\def\rootFx{7.5 + 3}
\def\rootFy{\rootTy}

\node[vertex,fill = black!50!] (rootF) at (\rootFx, \rootFy) {};

\foreach \i [count=\t from 1] in {1, 3, 6, 8} {
	\node[vertex] (aF\t) at (\i * \r + \rootFx + 2*\r, \i * \b + \rootFy) {};
	\node[vertex, fill opacity = 0.5, fill = red] (sF\t) at (\i * \r + \Lr + \rootFx, \i * \b + \rootFy) {};
	\node[vertex] (aCR\t) at (-\i * \r + \rootFx - 2*\r, \i * \b + \rootFy) {};
}
\foreach \i [count=\t from 1] in {1, 3, 6} {
  \node[vertex] (osF\t) at (\i * \r + \Lr + \rootFx + 1, \i * \b + \rootFy) {};
  \draw (sF\t) to node[midway, above] {$\tau-\lambda$} (osF\t) ;
}
\foreach \i [count=\t from 1] in {2, 4, 7} {
	\node[vertex] (bF\t) at (\i * \r + \rootFx + 2*\r, \i * \b + \rootFy) {};
	\node[vertex] (bCR\t) at (-\i * \r + \rootFx- 2*\r, \i * \b + \rootFy) {};
}
\node (dotsF) at (5 * \r + \rootFx+ 2*\r, 5 * \b + \rootFy) {$\dots$};
\node (dotsCR) at (-5 * \r + \rootFx - 2*\r, 5 * \b + \rootFy) {$\dots$};

\draw (rootF) -- node[left] {} (aF1);
\draw (aF1) -- node[left] {} (bF1);
\draw (bF1) -- node[left] {\footnotesize $\tau$} (aF2);
\draw (aF2) -- node[left] {\footnotesize $\gamma + 1$} (bF2);
\draw (bF2) -- node[left] {\footnotesize $\tau$} (dotsF);
\draw (dotsF) -- node[left] {\footnotesize $\tau$} (aF3);
\draw (aF3) -- node[left] {\footnotesize $\gamma + 1$} (bF3);
\draw (bF3) -- node[left] {\footnotesize $\tau$} (aF4);

\draw (rootF) -- node[right] {} (aCR1);
\draw (aCR1) -- node[right] {} (bCR1);
\draw (bCR1) -- node[right] {\footnotesize $\tau$} (aCR2);
\draw (aCR2) -- node[right] {\footnotesize $\gamma + 1$} (bCR2);
\draw (bCR2) -- node[right] {\footnotesize $\tau$} (dotsCR);
\draw (dotsCR) -- node[right] {\footnotesize $\tau$} (aCR3);
\draw (aCR3) -- node[right] {\footnotesize $\gamma + 1$} (bCR3);
\draw (bCR3) -- node[right] {\footnotesize $\tau$} (aCR4);

\foreach \i [count = \l] in {4.5, 10.5} {
	\foreach \j [count = \t] in {1} {
		\node (tau\l\t) at (\i, \rootFy + \j * 0.5 * \b) {\footnotesize $\tau$};
	}
	\foreach \j [count = \t] in {3} {
		\node (gamma\l\t) at (\i, \rootFy + \j * 0.5 * \b) {\footnotesize $\gamma + 1$};
	}	
}

\draw (sF4) -- node[above = 0, midway] {\footnotesize $\gamma+1$} (aF4);
\draw (sCL4) -- node[above = 0, midway] {\footnotesize $\gamma+1$} (aCR4);
\foreach \i in {1,2,3}{
  \draw (sF\i) -- node[above = 0, midway] {\footnotesize $\tau -\lambda$} (aF\i);
  \draw (sCL\i) -- node[above = 0, midway] {\footnotesize $\tau - 2\lambda$} (aCR\i);
}

\draw (sF4) to [out=90, in=90] node[below= 0, pos=0.5] {$\lfloor\frac{\tau + \gamma}{2}\rfloor + 1$} (sCL4);
\draw (sT4) to [out=90, in=90] node[below = 0, pos=0.5] {$\lfloor\frac{\tau + \gamma}{2}\rfloor + 1$} (sCL4);


\foreach \i/\name/\bendred/\bendblue/\value [count = \t] in {-3/v_x/25/10/\lambda, -1.5/v_y/50/13/, 3/v_z/35/20/}{
	\node[vertex, inner sep=1pt] (var\t) at (7.5 + \i, -5) {\footnotesize $\name$};
	\node[vertex, inner sep=1pt] (varbar\t) at (7.5 + \i, -5.7) {\footnotesize $\overline{\name}$};
	\node[draw, rounded corners, fit = (var\t) (varbar\t)] (rectvar\t) {};
	\draw[opacity = 0.5, red] (rectvar\t) to [bend right = \bendred] node[above = 0, pos=0.05] {\color{black} \footnotesize $\value$} (sF\t);
	\draw[opacity = 0.5, blue] (rectvar\t) to [bend left = \bendblue] node[above right= 0, pos=0.05] {\color{black} \footnotesize $\value$} (sT\t);
}
\draw[opacity = 0.5, green] (var1) to [bend right = 10] node[above = 0, pos=0.05] {\color{black} \footnotesize $\lambda$} (sCL1);
\draw[opacity = 0.5, green] (var2) to [bend right = 10] node[above = 0, pos=0.2] {} (sCL1);
\draw[opacity = 0.5, green] (var3) to [bend left = 10] node[above = 0, pos=0.1] {} (sCL2);

\begin{scope}[xshift=-0.5cm]
\draw[blue!20!black,fill=orange!20, opacity = 0.2, rounded corners=10,thick]
     (\rootBLx - 1, \rootBLy + 0.5) rectangle (2,-6.5);
\node[opacity = 1] at (\rootBLx - 0.5, \rootBLy - 6) {\large \color{black} $B_L$};
\end{scope}

\begin{scope}[xshift=0.5cm]
\draw[blue!20!black,fill=black!20, opacity = 0.2, rounded corners=10,thick]
     (\rootBRx + 1, \rootBRy + 0.5) rectangle (\rootBRx - 3, -6.5);
\node[opacity = 1] at (\rootBRx + 0.5, \rootBRy - 6) {\large \color{black} $B_R$};
\end{scope}

\draw[blue!20!black,fill=green!20, opacity = 0.2, rounded corners=10,thick]
     (7, -1.5) rectangle (8, 5);
\node[opacity = 1] at (7.5, 4.6) {\large \color{black} $C$};
\node at (7.5, -6.5) {Variable vertices};
\draw[opacity = 0.3] (sL1) to node[left, pos=0.5] {1} (osT2);
\draw[opacity = 0.3] (sL2) to node[left, pos=0.3] {1} (osT2);
\draw[opacity = 0.3] (sL3) to node[above = 0, pos=0.5] {1} (varbar1);

\draw[opacity = 0.5] (sL2) to [bend left = 10] node[above = 0, pos=0.09] {$2 \gamma + 2$} (sR2);

\end{tikzpicture}
}
\caption{Illustration of $(H, \omega)$.
  Centered at the top is the bottleneck sequence $B(T, C, F)$.
  The vertices of $X$ are in purple (left), and the vertices of $Y$ are in yellow (right).
  The edges incident to the variable vertices that are drawn in blue, green, red all have weight $\lambda$.
  Not to overburden the figure, we have only drew \emph{some} edges of the construction.
  Only one edge of the matching between $X$ and $Y$ is depicted, and the paddings of $\overline{v_x}$ and of $\overline{t_2}$ are (partially) represented (weight-1 edges toward $X$).
  The clause corresponding to the bottommost vertex of $C$ contains $x, y$ and some other variable (not shown), while that of the second bottommost vertex of $C$ contains $z$ (and two other variables).
  The roots of bottlenecks are in gray.
  The leftmost and rightmost gray vertices are the only two vertices of weight less than $\tau+\gamma+1$ (namely $\tau$).
 }
\label{fig:first-red}
\end{figure}

\medskip

We check that the weight padding works as intended.

\begin{lemma}\label{lem:saturation}
Every vertex of $H$ has weight at~least $\tau + \gamma + 1$, except two vertices.
\end{lemma}
\begin{proof}
  By construction, all the vertices with weight less than $\tau + \gamma + 1$ in $H'$ have weight exactly $\tau + \gamma + 1$ in~$H$, while the vertices with weight at~least $\tau + \gamma + 1$ in $H'$ have kept the same weight in~$H$.
  We shall just check the property for vertices in $V(B_L) \cup V(B_R)$.

  Note that every vertex of the spines $P(B_R), P(B_L)$ except the two roots have weight at least $\tau + \gamma + 1$.
  In particular, the last vertices of $P(B_R), P(B_L)$ have weight $2\tau -\gamma - 1$ and this is bigger than $\tau + 2\gamma +3$ since $3\gamma +4 < \tau$ from \Cref{def:gamma-lambda}.
  Furthermore, the vertices in $X$ and $Y$ have an attachment of weight $\tau - \gamma - 1$ and are incident to an edge of the matching between $X$ and $Y$ of weight $2\gamma + 2$.
  Hence any vertex in $X \cup Y$ has weight at least $\tau - \gamma - 1 + 2\gamma + 2 = \tau + \gamma + 1$.

 In conclusion, every vertex of $H$ has weight at~least $\tau + \gamma + 1$, but the roots of~$B_L$ and~$B_R$, which have weight $\tau$.
\end{proof}

\subsection{Preparatory lemmas}

We will make use of the following two lemmas.

\begin{lemma}\label{lem:clause-surrounding}
For any $(\tau + \gamma)$-balancing order $\prec$ of $H(\varphi)$, for every clause $c = x \lor y \lor z$ of~$\varphi$, $\{v_x, v_y, v_z\}$ surrounds~$v_c$.
\end{lemma}
\begin{proof}
Vertex $v_c$ has two attachments of weight $\tau - 2\lambda$.
\Cref{lem:Forcing-the-order,lem:bottle-order} imply that $v_c$ is surrounded by the two vertices it is attached to in~$B(T,C,F)$.
Assume for the sake of contradiction that $v_x$, $v_y$, and $v_z$ are all smaller (resp. larger) than $v_c$.
Then the left weight (resp. right weight) of $v_c$ is at least
$$
(\tau - 2\lambda) + \weight(v_xv_c) + \weight(v_yv_c)+ \weight(v_zv_c)
= \tau + \lambda \geqslant \tau + \gamma + 1;
$$
a~contradiction. 
\end{proof}

\begin{lemma}\label{lem:well-defined-order}
For any $(\tau+\gamma)$-balancing order $\prec$ of $H(\varphi)$, for any variable $x$ of $\varphi$, we either have $v_x \prec C$ or $C \prec v_x$. 
\end{lemma}
\begin{proof}
By \Cref{lem:Forcing-the-order}, either $T \prec C \prec F$ or $F \prec C \prec T$ holds.
Up to reversing the order, we can assume without loss of generality that $T \prec C \prec F$.

For each $i \in [n]$, vertex $t_i$ is incident to two edges of weight $\tau - \lambda$: $t_i \overline{t_i}$ and its attachment in~$B(T,C,F)$.
Since $2\tau - 2 \lambda > \tau + \gamma$, \Cref{lem:P3-order} ensures that $t_i$ is surrounded by the other endpoints of both edges.
We thus claim that $\{v_{x_i}, \overline{v_{x_i}}\}$ surrounds~$t_i$, where $x_i$ is the $i$-th variable of~$\varphi$.
Indeed if $\{v_{x_i}, \overline{v_{x_i}}\} \prec t_i$ (resp.~$t_i \prec \{v_{x_i}, \overline{v_{x_i}}\}$), the left weight (resp. right weight) of $t_i$ is at least $\tau - \lambda + 2 \lambda = \tau + \lambda > \tau + \gamma$; a~contradiction.
Similarly $\{v_{x_i}, \overline{v_{x_i}}\}$ surrounds~$f_i$.

As $t_i \prec C \prec f_i$, for $\{v_{x_i}, \overline{v_{x_i}}\}$ to both surround $t_i$ and $f_i$, it holds that $v_{x_i}$ is smaller than~$t_i$ or is larger than~$f_i$.
Therefore we indeed have $v_{x_i} \prec C$ or~$C \prec v_{x_i}$.
\end{proof}

\subsection{Correctness of the reduction}

We can now show that our reduction performs as announced.

\begin{lemma}\phantomsection\label{lem:bo-to-sat}
If $H(\varphi)$ admits a $(\tau + \gamma)$-balancing order, then $\varphi$ is satisfiable.
\end{lemma}
\begin{proof}
  Let $\prec$ be a~$(\tau + \gamma)$-balancing order of $H(\varphi)$.
  Let the valuation $\mathcal A$ set $x$ to true if and only if~$v_x \prec C$.
  This is well defined by \Cref{lem:well-defined-order}.
  By~\cref{lem:clause-surrounding}, for every clause $c = x \lor y \lor z$, vertex $v_c$ is surrounded by $\{v_x, v_y, v_z\}$.
  Hence $\mathcal A$ sets within $\{x, y, z\}$ at least one variable to true and at least one variable to false, thus satisfies~$c$.
\end{proof}

\begin{lemma}\label{lem:sat-to-bo}
If $\varphi$ is satisfiable, then $H(\varphi)$ admits a~$\tau$-balancing order.
\end{lemma}
\begin{proof}
We first give a~fixed $\tau$-balancing order $\prec_1$ of $H[V(B_L) \cup V(B_R)]$ that does not rely on $\varphi$ being satisfiable.
Then we give a~$\tau$-balancing order $\prec_2$ of $H - (V(B_L) \cup V(B_R))$.
This order is based on a~truth assignment satisfying~$\varphi$.
It will remain to argue that there is a~$\tau$-balancing order extending both $\prec_1$ and $\prec_2$.
This is done by proving that any extension of $\prec_1$ to $H$ keeps the left and right weights of vertices in $V(B_L) \cup V(B_R)$ small enough, and by indicating in which order each vertex of~$X$ should appear relatively to their unique neighbor in~$V(H) \setminus (V(B_L) \cup V(B_R))$.

\medskip

\textbf{Construction of $\bm{\prec_1}$.}
For $\prec_1$ fix any reverse order on the bottleneck $B_L$, followed by any direct order on the bottleneck $B_R$.
Observe that, in $H[V(B_L) \cup V(B_R)]$, every vertex $v$ outside $X \cup Y$ satisfies $\Delta_{\prec_1}(v) \leqslant \tau$ (see the proof of~\cref{lem:direct-order}, or~\cref{fig:bottleneck}).
Consider the vertex $x \in X$.
It has two neighbors: one in $V(B_L)$ and one, say, $y$ in $Y$.
The left weight of $x$ is that of its attachment $\tau - \gamma - 1 \leqslant \tau$, whereas the right weight of $x$ is $\weight(xy) = 2\gamma + 2 \leqslant \tau$.
The situation is symmetric for the vertices of~$Y$.
Hence $\prec_1$ is a~$\tau$-balancing order of~$H[V(B_L) \cup V(B_R)]$.

\medskip

\textbf{Construction of $\bm{\prec_2}$.}
Let $\mathcal A$ be a~variable assignment satisfying $\varphi$.
We build the order~$\prec_2$ on $H - (V(B_L) \cup V(B_R))$ in the following way:
\begin{compactitem}
	\item we first have all vertices $\overline{t_i}$ for $i \in [n]$,
	\item followed by all the vertices $v_x$ (resp.~$\overline{v_x}$) such that $x$ is set to true (resp.~false) by $\mathcal A$,
	\item followed by any fixed direct order of $B(T, C, F)$,
	\item followed by the vertices $v_x$ (resp.~$\overline{v_x}$) such that $x$ is set to false (resp.~true) by $\mathcal A$,
	\item followed by all vertices $\overline{f_i}$ for $i \in [n]$.
\end{compactitem}

\medskip

We verify that $\prec_2$ is a~$\tau$-balancing order in $H - (V(B_L) \cup V(B_R))$.
We let $u$ be a vertex of $H - (V(B_L) \cup V(B_R))$, and prove that $\Delta_{\prec_2}(u) \leqslant \tau$.

\medskip

\textbf{Case $\bm{u}$ of the form $\bm{\overline{t_i}}$ or $\bm{\overline{f_i}}$ for $\bm{i \in [n]}$.}
The vertex $u$ is incident to a~single edge of weight $\tau - \lambda \leqslant \tau$.

\medskip

\textbf{Case $\bm{u}$ of the form $\bm{v_x}$ or $\bm{\overline{v_x}}$ for some variable $\bm{x}$.}
Let $x_i$ be the $i$-th variable of $\varphi$ and $c_1, c_2, c_3, c_4$ be the four clauses of $\varphi$ in which $x_i$ appears.
Then $u = v_{x_i}$ (resp.~$u = \overline{v_{x_i}}$) is incident to six edges (resp. two edges) each of weight $\lambda$ to $t_i, f_i, v_{c_1}, v_{c_2}, v_{c_3}, v_{c_4}$ (resp.~$t_i, f_i$), thus $\Delta_{\prec_2}(u) \leqslant 6 \lambda \leqslant \tau$. 

\medskip

\textbf{Case $\bm{u}$ of the form $\bm{t_i}$ or $\bm{f_i}$ for some $\bm{i \in [n]}$.}
The vertex $t_i$ has exactly four neighbors: the vertex $\overline{t_i}$, the vertex $v$ it is attached to in $B(T,C,F)$, and the two vertices $v_{x_i}$ and $\overline{v_{x_i}}$.
By definition of $\prec_2$, we have $\overline{t_i} \prec_2 t_i \prec_2 v$, and either $v_{x_i} \prec_2 t_i \prec_2 \overline{v_{x_i}}$ or $\overline{v_{x_i}} \prec_2 t_i \prec_2 v_{x_i}$.
Since $\weight(t_i \overline{t_i}) = \weight(t_i v) = \tau - \lambda$, and $\weight(t_i v_{x_i}) = \weight (t_i \overline{v_{x_i}}) = \lambda$, the left weight and right weight of $t_i$ are both equal to $(\tau-\lambda)+\lambda=\tau$.
The situation is symmetric for $f_i$.

\medskip

\textbf{Case $\bm{u}$ of the form $\bm{v_c}$ for a~clause $\bm{c}$.}
Let $c = x \lor y \lor z$.
The vertex $v_c$ is adjacent to exactly five vertices: two vertices, say $s_1$ and $s_2$, it is attached to in the bottleneck sequence $B(T,C,F)$ via an edge of weight $\tau - 2 \lambda$, and $v_x, v_y, v_z$.
We have $s_1 \prec_2 v_c \prec_2 s_2$, and since $\varphi$ is satisfied by $\mathcal A$, there is at most two vertices among $v_x, v_y, v_z$ to the right of $v_c$, and at most two to its left.
Since $\weight (v_cv_x) = \weight (v_cv_y) = \weight (v_cv_z) = \lambda$, the left weight and the right weight of $v_c$ are at most $(\tau - 2\lambda) + 2\cdot \lambda = \tau$.

\medskip

\textbf{Case $\bm{u}$ in one of the spines of $\bm{B(T, C, F)}$.}
Since $u$ has no neighbors outside of the bottleneck sequence, the order $\prec_2$ acts as $\prec_\rightarrow^{\text{seq}}$, and we conclude by~\Cref{lem:direct-order}.

\medskip

\textbf{Extending $\bm{\prec_1}$ and $\bm{\prec_2}$.}
Note that every vertex of $X$ has right weight and left weight at~most $\tau - 1$ in $\prec_1$, and is adjacent to exactly one vertex outside of $V(B_L) \cup V(B_R)$ with an edge of weight~1.
Hence if $\prec$ is any total order extending $\prec_1$ to $V(H)$, every $v \in V(B_L) \cup V(B_R)$ satisfies $\Delta_\prec(v) \leqslant \tau$.
Indeed, vertices in $ (V(B_L) \cup V(B_R)) \setminus X$ have no neighbors outside $V(B_L) \cup V(B_R)$.

Let $u \in V(H) \setminus (V(B_L) \cup V(B_R))$, and let $s$ be its weight towards $H - (V(B_L) \cup V(B_R))$.
If $s \geqslant \tau + \gamma + 1$, then $u$ is not adjacent to $V(B_L) \cup V(B_R)$, so any extension $\prec$ of $\prec_2$ to $H$ keeps $\Delta_\prec(u) \leqslant \tau$.
Otherwise, let $s_L$ and $s_R$ be the left weight and right weight of $u$, respectively (w.r.t.~$\prec_2$).
Note that $s_L + s_R = s$.
The vertex $u$ has exactly $\tau + \gamma + 1 - s$ neighbors in~$X$, $x_1, \ldots, x_{\tau + \gamma + 1 - s}$ via edges of weight~1.

If $s_R > \gamma+1$, we simply set $x_i \prec u$ for every $i \in [\tau + \gamma + 1 - s]$.
The left weight of $u$ in $H$ ordered by~$\prec$ is at~most $s_L+\tau + \gamma + 1 - s = \tau + \gamma + 1 - s_R \leqslant \tau$.
If instead $s_R \leqslant \gamma+1$, we set $x_i \prec u$ for every $i \in [1,\tau - s_L]$, and $u \prec x_i$ for every $i \in [\tau - s_L + 1,\tau + \gamma + 1 - s]$.
This is well defined since $\tau-s_L \leqslant \tau+\gamma+1-s$.
The left weight of $u$ in $H$ ordered by~$\prec$ is at~most $s_L + (\tau - s_L) \leqslant \tau$, and its right weight is
$s_R + (\gamma - s + s_L + 1) =  \gamma + 1 \leqslant \tau$.
\end{proof}

\cref{thm:ldb} is a~direct consequence of~\cref{lem:bo-to-sat,lem:sat-to-bo}.
The reason we ``padded the degree'' in $H(\varphi)$ will become apparent in the next section.
We will observe that when $\varphi$ is unsatisfiable, not only no linear order can ``balance'' the degrees, but no tree can either.

\subsection{From linear orders to trees}\label{sec:lin-ord-to-trees}

As mim-width is defined via branch decompositions, we adapt the balancing order problem to trees.
Consider an edge-weighted graph $(H, \weight)$, and a~tree $T$ such that there exists a~bijective map $f \colon V(H) \rightarrow V(T)$.
Note that $(T, f)$ is \emph{not} a~branch decomposition of $H$ for two reasons: vertices of $H$ are mapped to \emph{all} nodes~of $T$ and not merely its leaves, and $T$ is not necessarily a~ternary tree (nor a~rooted binary tree).

Each edge $e$ of $T$ defines a~cut of $H$, which we denote $(A_e, B_e)$, where $A_e$ is the preimage by~$f$ of one connected component of $T-e$, and $B_e$, of the other component.
We say that $(T, f)$ is a~\emph{$\tau$-balancing tree of $(H, \weight)$} if for any vertex $v \in V(H)$, for any edge $e \in E(T)$ incident to~$f(v)$, the sum of the weights of edges (in $E(H)$) incident to~$v$ in the cut $(A_e, B_e)$ is at~most~$\tau$.

\defparproblem{\tdb}{An edge-weighted graph~$(H, \weight)$ and a~non-negative integer~$\tau$.}{$\tau$}{Does $(H, \weight)$ admit a~$\tau$-balancing tree?}

Note that any graph with a~$\tau$-balancing order also admits a~$\tau$-balancing tree, with $T$ being the path of length $|V(H)|$, and $f$ mapping the vertices of~$H$ along $T$ in the $\tau$-balancing order.

\begin{theorem}\label{thm:tdb}
  Given an edge-weighted graph $(H, \weight)$ promised to satisfy either one of 
  \begin{compactitem}
  \item $(H, \weight)$ admits a~$\tau$-balancing order or 
  \item $(H, \weight)$ does not admit a~$(\tau + \gamma)$-balancing tree,
  \end{compactitem}
  deciding which outcome holds is \NP-hard.
\end{theorem}
\begin{proof}
  In the previous reduction, since all the vertices of $H(\varphi)$ but two have degree at~least $\tau + \gamma + 1$, any $(\tau + \gamma)$-balancing tree $(T,f)$ is such that $T$ has at most two leaves, and so $T$ is a~path.
  So in the case when $H(\varphi)$ has no $(\tau + \gamma)$-balancing order, $H(\varphi)$ has no $(\tau + \gamma)$-balancing tree.
\end{proof}

\section{\textsc{Degree Balancing} to \textsc{Linear Mim-Balancing/Tree Sim-Balancing}}\label{sec:degree-to-matching}

In this section, we show how to transfer the degree requirement of \textsc{Degree Balancing} to the induced-matching requirement of \textsc{Mim-} and \textsc{Sim-Balancing}.


\subsection{The \textsc{Mim-Balancing} and \textsc{Sim-Balancing} problems}\label{sec:im-balancing}

A~\emph{partitioned graph} is a~pair $(G, \mathcal S)$ where $G$ is a~graph and $\mathcal S$ is a~partition of $V(G)$. 
A~\emph{\tmap} of a~partitioned graph $(G,\mathcal S)$ is a~pair $(T,f)$ where $T$ is a~tree and $f \colon \mathcal S \rightarrow V(T)$ is a~bijection from the parts of $\mathcal S$ to the vertices of~$T$. 
When $T$ is a path, we may call $(T, f)$ a~\emph{\pmap of~$\mathcal S$}.

We say that a cut $(A,B)$ of $G$ is an \emph{$\mathcal S$-cut} if each set in $\mathcal S$ is included in either $A$ or $B$.
Each edge $e \in E(T)$ in a~tree mapping $(T, f)$ of $(G, \mathcal S)$ defines an $\mathcal S$-cut $(A_e, B_e)$ of $G$: the union of the parts mapped to each component of $T - e$.
The \emph{sim-value (resp. mim-value) of a~\tmap $(T, f)$ of $(G,\mathcal S)$} is the maximum taken over every edge $e \in E(T)$ of the maximum size of an induced (resp. semi-induced) matching between $A_e$ and $B_e$. 
The \emph{sim-balancing (resp. mim-balancing) of $(G, \mathcal{S})$} is the minimum sim-value (resp. mim-value) among all possible \tmaps of $(G, \mathcal{S})$. 
Similarly, the \emph{linear sim-balancing (resp. mim-balancing)} is the minimum sim-value (resp. mim-value) among \pmaps.

\defparproblem{\tsb (resp.~\tmb)}{A~partitioned graph~$(G,\mathcal S)$ and a~non-negative integer~$\tau$.}{$\tau$}{Does $(G,\mathcal S)$ admit a~\tmap $(T, f)$ of sim-value (resp.~mim-value)~$\tau$?}

Note that even when $\mathcal S$ is the finest partition $\{\{v\}~:~v \in V(G)\}$, this problem is not exactly \textsc{Mim-Width}, as $f$ also maps vertices to internal nodes of $T$, and $T$ has no degree restriction.
\lmb (or \lsb) is defined analogously except $T$ is forced to be a~path.
We may use \textsc{Mim-Balancing} to collectively refer to \lmb and \tmb; and similarly with \textsc{Sim-Balancing}.

At the end of this section, we will have established the following.

\begin{theorem}
  Let $\tau, \gamma$ be natural numbers satisfying \Cref{def:gamma-lambda} and $\gamma > 50$.
  Given partitioned graphs $(G,\mathcal S)$ such that:
  \begin{compactitem}
  \item the linear mim-balancing of $(G, \mathcal{S})$ is at most $\tau + 50$, or
  \item the sim-balancing of $(G, \mathcal{S})$ is at least $\tau + \gamma$,
  \end{compactitem}
  deciding which of the two outcomes holds is \NP-hard.
\end{theorem}

\subsection{Encoding \textsc{Degree Balancing} in \textsc{Mim/Sim-Balancing}}

Let $(H, \weight)$ be an instance of \tdb with positive and integral weights.
We build an instance of \tmb $G := G(H, \weight), \mathcal S := \mathcal{S}(H, \weight)$, as follows.

\medskip

\textbf{Construction of $\bm{(G,\mathcal S)}$.}
For every vertex $u \in V(H)$ and every $v \in N_H(u)$, we add an independent set $I(u, v)$ of size $\weight(uv)$ to~$G$.
For each vertex $u \in V(H)$, we set~$$S(u) := \bigcup_{v \in N_H(u)} I(u, v).$$
Each $S(u)$ will remain an independent set in~$G$.
The partition $\mathcal S$ is simply defined as $\{S(u)~:~u \in V(H)\}$.

We finish the construction by adding two kinds of edges in $G$, \emph{matching edges} and \emph{dummy edges}. 
For every pair of disjoint edges $uv$ and $xy$ of $H$, we add an edge between every vertex of $I(u, v)$ and every vertex of $I(x, y)$.
All these edges are called \emph{dummy}.
For every $uv \in E(H)$, we add a~maximum (perfect) induced matching between $I(u,v)$ and $I(v,u)$.
All these edges are called \emph{matching edges}.
Observe that $\weight(uv)=\weight(vu)$ ($H$ is undirected), hence $|I(u,v)|=|I(v,u)|$ and the matching between $I(u,v)$ and $I(v,u)$ is indeed perfect.  
This concludes the construction of $(G, \mathcal S)$; see~\cref{fig:matching-and-dummy} for an illustration of the adjacencies between some $S(u)$ and $S(v)$.

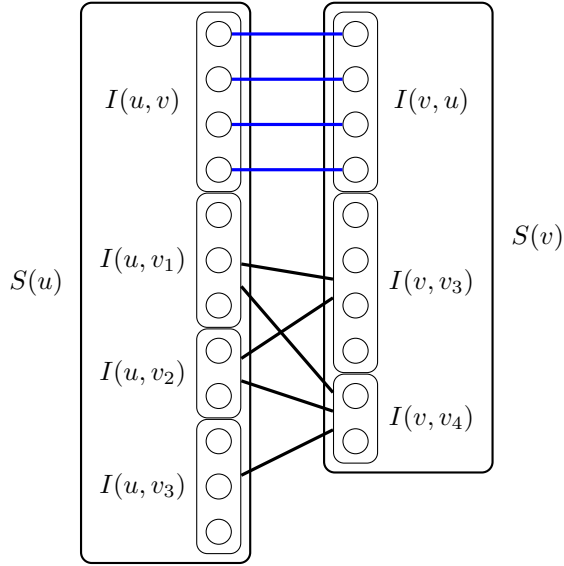
\begin{figure}[h!]
  \centering
  \begin{tikzpicture}[
      vertex/.style={draw, circle, minimum size = 0.3cm},
      matching/.style={blue, very thick},
      dummy/.style={very thick},
    ]
    \def\s{0.6}
    \def\h{1}
    \def\n{12}
    \foreach \i in {1,...,\n}{
      \node[vertex] (a\i) at (0, - \i * \s) {} ;
    }
    \foreach \i/\j/\k/\l in {1/4/1/v,5/7/2/v_1,8/9/3/v_2,10/12/4/v_3}{
      \node[draw, rounded corners, fit=(a\i) (a\j)] (Iu\k) {} ;
      \node (tIu\k) at (-\h, - 0.5 * \i * \s - 0.5 * \j * \s) {$I(u,\l)$} ; 
    }
    \node[draw, thick, rounded corners, fit=(Iu1) (Iu4) (tIu1) (tIu4)] (Su) {} ;
    \node at (- 2.4 * \h, - 0.5 * \s - 0.5 * \n * \s) {$S(u)$} ;

    \def\m{10}
    \foreach \i in {1,...,\m}{
      \node[vertex] (b\i) at (3 * \s, - \i * \s) {} ;
    }
    \foreach \i/\j/\k/\l in {1/4/1/u,5/8/2/v_3,9/10/3/v_4}{
      \node[draw, rounded corners, fit=(b\i) (b\j)] (Iv1\k) {} ;
      \node (tIv1\k) at (3 * \s + \h, - 0.5 * \i * \s - 0.5 * \j * \s) {$I(v,\l)$} ; 
    }
    \node[draw, thick, rounded corners, fit=(Iv11) (Iv13) (tIv11) (tIv13)] (Sv1) {} ;
    \node at (3 * \s + 2.4 * \h, - 0.5 * \s - 0.5 * \m * \s) {$S(v)$} ;

    \foreach \i in {1,...,4}{
      \draw[matching] (a\i) -- (b\i) ;
    }
    \foreach \i/\j in {Iu2/Iv12, Iu2/Iv13, Iu3/Iv12, Iu3/Iv13, Iu4/Iv13}{
      \draw[dummy] (\i) -- (\j) ;
    }
  \end{tikzpicture}
  \caption{Adjacencies between $S(u)$ and $S(v)$.
    In this example, $u$ has four neighbors $v, v_1, v_2, v_3$, and $v$ has three neighbors $u, v_3, v_4$.
    The matching edges are in blue, the dummy edges are in black (edges between two boxes represent bicliques).
    Notice the non-edges between $I(u,v_3)$ and $I(v,v_3)$.}
  \label{fig:matching-and-dummy}
\end{figure}

We notice that the configuration of the figure actually implies that $uvv_3$ is a~triangle in~$H$, which does not happen in graphs $H$ produced by the previous reduction.
However, we will not use that $H$ is triangle-free in the current section, and \cref{fig:matching-and-dummy} shows the general behavior between $S(u)$ and $S(v)$.
(For triangle-free graphs $H$, if $uv \in E(H)$, then there would instead be a~biclique between $S(u) \setminus I(u,v)$ and $S(v) \setminus I(v,u)$, and if $uv \notin E(H)$, $I(u,v)$, $I(v,u)$, and the matching edges in between them would simply not exist.) 

\subsection{Preparatory lemmas}

\renewcommand{\part}{\text{part}}

We will now prove some facts about the (semi-)induced matchings of the $\mathcal S$-cuts of $G$.

\begin{lemma}\label{lem:no-dummy-edge-S-cut}
	Let $(A, B)$ be an $\mathcal S$-cut of $G$.
	If there is no dummy edge between a vertex of $I(u,v)\subseteq A$ and one from $I(x,y)\subseteq B$, then $u=y$ or $v=x$ or $v=y$.
\end{lemma}
\begin{proof}
	If there is no edge between $I(u,v)$ and $I(x,y)$ in $G$, then by construction $u=x$ or $u=y$ or $v=x$ or $v=y$.
	But since $(A,B)$ is an $\mathcal S$-cut that separates $I(u,v)$ from $I(x,y)$, we have $u\neq x$.
\end{proof}

\begin{lemma}\label{lem:matching-edges-are-clean}
  Let $(A, B)$ be an $\mathcal S$-cut of $G$.
  Assume there exists a semi-induced matching $M := \{e_1, \dots, e_m\}$ in $G$ between $A$ and $B$ containing matching edges only.
  Then a~single part of~$\mathcal S$ covers all the edges of $M$. 
\end{lemma}
\begin{proof}
  Let us denote by $a_i \in A, b_i \in B$ the two endpoints of $e_i$.
  Since the edges of $M$ are \emph{matching edges} they are between pairs of sets of the form $I(u, v)$ and $I(v, u)$.
  Hence, we denote by $I(u_i, v_i) \subseteq A$ the set containing $a_i$ and by $I(v_i, u_i) \subseteq B$ the set containing $b_i$.
  Recall that $I(u_i, v_i) \subseteq S(u_i)$ and  $I(v_i, u_i) \subseteq S(v_i)$.
  Thus, as $(A, B)$ is an $\mathcal S$-cut, $S(u_i) \subseteq A$ and $S(v_i) \subseteq B$.
  \begin{claim}\label{cl:proof-matching-edges-are-clean}
    For every $i,j\in [m]$, we have $u_i = u_j$ or $v_i = v_j$.
  \end{claim}
  \begin{proofofclaim}
    Since $M$ is semi-induced, the vertex $a_i$ is not adjacent to $b_j$, thus $a_ib_j$ is not a~dummy edge.
    By \Cref{lem:no-dummy-edge-S-cut}, we have $u_i = u_j$, $v_i = u_j$ or $v_i=v_j$.
    Since $u_j \in A$ and $v_i \in B$, we have $v_i \neq u_j$.
    Hence, either $u_i = u_j$ or $v_i = v_j$.
  \end{proofofclaim}

  Applying \Cref{cl:proof-matching-edges-are-clean} to every pair $e_1, e_i$ for $i \in [m]$, we get that $u_i = u_1$ or $v_i = v_1$ for every $i\in [m]$.
  If $u_i = u_1$ for all $i\in [m]$, or $v_i = v_1$ for all $i\in [m]$, then $S(u_1)$ or $S(v_1)$ covers $M$ and the lemma holds.
  Hence, assume there exist $i,j\in [m]$ such that $u_i = u_1$ and $u_j\neq u_1$ and $v_i \neq v_1$ and $v_j = v_1$.
  Applying \Cref{cl:proof-matching-edges-are-clean} to $e_i, e_j$ we get that $u_i = u_j$ or $v_i = v_j$.
  This implies that $u_j = u_1$ or $v_i = v_1$; a~contradiction to the fact $u_j \neq u_1$ and $v_i\neq v_1$.
\end{proof}

\begin{lemma}\label{lem:white-edges-restriction}
  Let $(A, B)$ be an $\mathcal S$-cut of~$G$.
  If $M := \{a_1b_1, \dots, a_tb_t\}$ is a semi-induced matching in $G$ between $A$ and $B$ only made of dummy edges, and all the vertices $a_i$ lie in the same part of $\mathcal S$ included in $A$, then $t \leqslant 6$.
\end{lemma}
\begin{proof}
  Let $S(u)$ be the part of $\mathcal S$ including $\{a_1, \ldots, a_t\}$.
  Let us denote by $I(u, v_i) \subseteq A$ the set containing $a_i$ and by $I(x_i, y_i) \subseteq B$ the set containing $b_i$.
  By definition of the dummy edges, for every $i \in [t]$,  we have
  \begin{align}\label{eq:dummy_cond}
    v_i \neq x_i \text{ and } u \neq y_i \text{ and }v_i \neq y_i.
  \end{align}
  
  Let $A_M = \{a_1,\dots,a_t\}$ and $B_M = \{b_1,\dots,b_t\}$.
  Observe that for every $i\neq j \in [m]$, since $M$ is a semi-induced matching $a_ib_j$ is not a edge of $G$, by \Cref{lem:no-dummy-edge-S-cut}, we have 
  $u = y_j$ or $v_i = x_j$ or $v_i = y_j$, but from  Condition~(\ref{eq:dummy_cond}) we know that $u\neq x_j$.
  Thus, for every $i\neq j \in [m]$, we have
  \begin{align}\label{eq:no_dummy_cond}
  	v_i = x_j \text{ or } v_i = y_j.  
  \end{align}
  We first prove that any part of $\mathcal S$ contains at most two vertices of $B_M$.
  Indeed, assume (without loss of generality) that $b_1, b_2$ and $b_3$ are all in a single part of $\mathcal S$, i.e., $x := x_1 = x_2 = x_3$.
  From (\ref{eq:no_dummy_cond}), we have $v_i = x$ or $v_i = y_j$, for every $i\neq j \in [3]$.
  However, by~Condition~(\ref{eq:dummy_cond}), we get that $v_i \neq x$.
  Thus, only one disjunct remains: $v_i = y_j$.
  But then $v_1 = y_2 = v_3 = y_1$, and $v_1 = y_1$ contradicts~(\ref{eq:dummy_cond}).
  
  Thus, any part of $\mathcal S$ contains at most two vertices of $B_M$.
  And in particular, for every $i \in [t]$, $|S(v_i) \cap B_M| \le 2$.
  For every $i \neq j$, we know from (\ref{eq:no_dummy_cond}) that $v_i = x_j$ or $v_i = y_j$.
  For a~fixed $i$, only two vertices of $B_M$ can satisfy the first disjunct, thus, we have $v_i = y_j$ for at least $t - 1 - 2 = t - 3$ of the indices $j \in [t] \setminus \{i\}$.
  
  Assume for the sake of contradiction that $t > 6$.
  We have $2 \cdot (t-3) > t$, so for every $i, j \in [t]$, $\{k ~:~ y_k = v_i\} \cap \{k ~:~ y_k = v_j\} \neq \emptyset$, which implies that $v_i = v_j$.
  Hence since there exist $i, k$ with $y_k = v_i$, and as $v_i = v_k$, we have $v_k = y_k$; contradicting~(\ref{eq:dummy_cond}).
\end{proof}

\begin{lemma}\label{lem:mim-width-local}
  Let $(A, B)$ be a $\mathcal S$-cut of $(G, \mathcal S)$ with a~semi-induced matching $\{e_1, \ldots, e_m\}$ between $A$ and $B$.
  Then at least $m - 50$ edges among $\{e_1, \ldots, e_m\}$ are matching edges.
\end{lemma}
\begin{proof}
Up to reordering, assume that $D := \{e_1, \ldots, e_t\}$ are the dummy edges of $\{e_1, \ldots, e_m\}$ for some $t \in [m]$.
We denote by $a_i \in A, b_i \in B$ the endpoints of $e_i$, and by $I(u_i, v_i) \subseteq A$ the set containing $a_i$ and by $I(x_i, y_i) \subseteq B$ the set containing $b_i$. 

By definition of the dummy edges, we have that for every $i \in [t]$,
\begin{align}
  \label{eq:dummy_cond2} v_i \neq x_i \text{ and } u_i \neq y_i \text{ and }v_i \neq y_i.
\end{align}

Let us consider the auxiliary directed graph~$\Aux$ where $V(\Aux):= D$, and $E(\Aux)$ contains the arc $(e_i, e_j)$ whenever $y_i = u_j$ or $v_i = x_j$.
By \Cref{lem:white-edges-restriction}, each part of~$\mathcal S$ contains at most 6 vertices of $\{a_1, \dots, a_t \} \cup \{b_1, \dots, b_t\}$.
Therefore, each of the disjuncts ($y_i = u_j$ or $v_i = x_j$) creates at most 6 outgoing arcs from~$e_i$.
Hence $\Aux$ has maximum outdegree at~most~12.
Thus the underlying undirected graph $J$ of $\Aux$ is 24-degenerate.
Thus $J$ admits an independent set $U$ of size $\left\lceil |D|/25 \right\rceil$.

Assume for the sake of contradiction that $t > 50$, hence that $|U| \geqslant 3$.
Without loss of generality, say that $e_1, e_2, e_3 \in U$.
Since $M$ is semi-induced, for every $i \neq j \in [t]$, $a_ib_j$ is not a dummy edge, thus by \Cref{lem:no-dummy-edge-S-cut} we have $v_i = x_j$ or $u_i = y_j$ or $v_i = y_j$.
But when $i$ and $j$ are restricted to $\{1, 2, 3\}$, there is no arc in $\Aux$ between $e_i$ and $e_j$, so $v_i \neq x_j$ and $u_i \neq y_j$. 
  Hence only the third disjunct can hold.
  Hence we have $v_1 = y_2 = v_3 = y_1$, and $v_1 = y_1$ contradicts (\ref{eq:dummy_cond2}).
\end{proof}

\subsection{Correctness of the reduction}

We can now show the correctness of the reduction.

\begin{lemma}\label{lem:simwidth2balanced}
If $(G, \mathcal S)$ admits a~\tmap~$(T, f)$ of sim-value~$t$, then $(H,\weight)$ admits a~$t$-balancing tree.
\end{lemma}
\begin{proof}
Consider a~\tmap~$(T, f)$ of sim-value at most~$t$.
We keep the same tree $T$ and define the map $f': V(H) \rightarrow V(T)$ with $f'(v) := f(S(v))$.
We will show that $(T, f')$ is a~$t$-balancing tree of~$H$.

Consider an edge $e \in E(T)$.
Let $(A_e^H, B_e^H)$ (resp.~$(A_e^G, B_e^G)$) be the cut in~$H$ (resp.~in~$G$) defined by $e$, such that for every $v \in V(H)$, $v \in A_e^H$ if and only if $S(v) \subseteq A_e^G$.
Note that $(A_e^G, B_e^G)$ is an $\mathcal S$-cut.
Consider a vertex $v \in V(H)$.
Up to swapping $A_e^H$ and $B_e^H$ (and $A_e^G$ and $B_e^G$ accordingly), we may assume that $v \in A_e^H$.
Consider $u_1, \ldots, u_p$ an enumeration of $N_H(v) \cap B_e^H$.
Now consider $M$, the set of all matching edges in $G$ going from $S(v)$ to $S(u_1) \cup \dots \cup S(u_p)$.
By construction $|M| = \sum_{i \in [p]} \weight(vu_i)$.

We prove that $M$ is an induced matching between $A_e^G$ and $B_e^G$.
Let us denote by $a_1b_1, \ldots, a_mb_m$ the edges of $M$ with $a_i \in S(v)$ for each $i \in [m]$.
Note that $a_i \in A_e^G$ and $b_i \in B_e^G$.
Since $M$ is made of matching edges, for every $i\in [m]$, there exists a vertex $x_i \in \{u_1,\dots,u_p\}$ such that $b_i \in I(x_i, v)$.
Two vertices $a_i$ and $a_j$ are not adjacent since $S(v)$ is an independent set.
By construction, $b_ib_j$ is not a dummy edge of $G$ since $b_i \in I(x_i, v)$ and $b_j \in I(x_j, v)$.  
This is also the case for $a_ib_j$ since $a_i \in I(v, x_i)$ and $b_j \in I(x_j, v)$, and it holds also for $a_jb_i$ by symmetry.
As every vertex of~$G$ is incident to at~most one matching edge, $M$ is indeed an induced matching in $G$ between $A_e^G$ and $B_e^G$.

Since the sim-value of $(T, f)$ is $t$, we have $|M| \leqslant t$, and so $\sum_{i \in [p]} \weight(vu_i) \leqslant t$.
This upper bound was shown for every $e \in E(T)$ and $v \in V(H)$, so $(T, f')$ is a~$t$-balancing tree of~$H$.
\end{proof}

\begin{lemma}\label{lem:balanced2sim-width}
If $H$ admits a $\tau$-balancing order, then there is a~\pmap~$(P, f)$ of~mim-value at~most~$\tau + 50$.

Moreover, for every cut $(A_e, B_e)$ induced by an edge $e\in E(P)$ and every semi-induced matching $M$ between $A_e$ and $B_e$, $M$ has at most $\tau$ matching edges and there exists $u\in V(H)$ such that $S(u)$ covers the matching edges of $M$.
\end{lemma}

\begin{proof}
Let $\prec$ be a~$\tau$-balancing order of~$H$.
Let us call $v_1 \prec \dots \prec v_n$ the vertices of~$H$.
We define $P := p_1 \ldots p_n$ as be the path of order~$n$, and $f$ as the map $S(v_i) \mapsto p_i$.
It remains to bound the mim-value of~$(P, f)$.

Consider any $i \in [n-1]$ and the edge $e = p_ip_{i+1} \in E(P)$,
and let $(A_e, B_e)$ the cut of~$G$ such that $S(v_1), \dots, S(v_i) \subseteq A_e$, and $S(v_{i+1}), \dots, S(v_n) \subseteq B_e$.
Let $M = \{e_1, \ldots, e_m\}$ be a~semi-induced matching between $A_e$ and $B_e$. By \Cref{lem:mim-width-local}, one can assume that $M' := \{e_1, \dots, e_{m - 50}\}$ contains only matching edges.
By~\Cref{lem:matching-edges-are-clean}, all the edges in~$M'$ are incident to a same part, say~$S(u)$.

This implies that all edges of $M'$ are of the form $a_jb_j$ with $a_j \in S(u)$ and $b_j$ in some $S(v_j)$ with $uv_j \in E(H)$.
By construction of $P$, we either have $\{v_1, \dots, v_{|M'|}\} \prec u$, or $u \prec \{v_1, \dots, v_{|M'|}\}$. 
In particular, $|M'|$ is at most the maximum between the left weight and the right weight of $u$, which is at most $\tau$.
Hence $|M| \leqslant \tau + 50$, and since this applies to any edge of $P$, the mim-value of~$(P, f)$ is at~most~$\tau + 50$. 
\end{proof}

\section{\textsc{Mim/Sim-Balancing} to \textsc{Linear Mim-Width/Sim-Width}}\label{sec:balancing-to-mim-width}

\newcommand{\gad}{\mathcal G}

The next reduction uses two constants $a := 45$ and $b:= 6 \tau (\tau + \gamma) + 1$.
With the announced values of $\tau = 1080$ and $\gamma = 135$, we have $b = 7873201$.
We remark that the value of~$b$ will not affect the linear mim-width upper bound nor the sim-width lower bound.
(The constant $b$ should simply be that large to make our proofs work.)

Let $(H, \weight)$ be an instance of \tdb where all edge weights are positive multiples of~$a$ and $H$ is triangle-free.
We build a graph $G^*$, such that if $H$ has a~$\tau$-balancing order, then the linear mim-width of $G^*$ is at most $\frac{a+1}{a} \tau + 107$; and if $H$ is has no $(\tau+\gamma)$-balancing tree, then the sim-width of $G^*$ is at~least $\tau+\gamma$.
We construct $G^*$ from the instance $G := G(H), \mathcal S := \mathcal{S}(H)$ of \tmb from the previous reduction.
Remember that $\mathcal S = \{S(u)~:~u \in V(H)\}$.

The main goal of this reduction is to obtain a graph $G^*$ whose sim-width and linear mim-width are related to the the sim-balancing and linear mim-balancing of $(G,\cal S)$, respectively.
Observe that we cannot simply set $G^*:=G$ since a~layout $(T,f)$ of~$G$ can scatter each $S(u)$ so that for each matching edge $xy$ of $G$, $x$ and $y$ are placed at leaves of~$T$ sharing a~neighbor in~$T$.
Consequently, the only cuts $(A,B)$ induced by the edges of~$T$ with a~matching edge between $A$ and $B$ are rather trivial ($A$ or $B$ is a~singleton).
Thus, the sim-width of $G$ could be uncontrollably smaller than the sim-balancing of~$(G,\cal S)$. 

To prevent this, we design a gadget $\gad(u)$ for each $u\in V(H)$ from $b$ copies of $S(u)$.
These gadgets ensure that any tree layout $(T,f)$ of $G^*$ of sim-value at~most $\tau+\gamma-1$ behaves similarly to a~tree mapping of $(G,\cal S)$ in the sense that for every $S(u)$, there is an edge $e$ of $T$ such that both sides of the induced cut $(A_e, B_e)$ contain a copy of $S(u)$. 
Using this property, we prove that the sim-width of $G^*$ is at~least the sim-balancing of $(G,\cal S)$.

The final lemma from each of the two last subsections prove~\cref{thm:main} (hence~\cref{thm:main-csq}), which we restate here.
\begin{theorem}
  Let $\tau, \gamma, a$ be natural numbers as previously defined.
  Given graphs $G$ such that either 
  \begin{compactitem}
  \item the linear mim-with of~$G$ is at~most $\frac{a+1}{a}\tau + 107$, or
  \item the sim-width of $G$ is at~least $\tau + \gamma + 1$,
  \end{compactitem}
  deciding which of the two outcomes holds is \NP-hard.
\end{theorem}
One can indeed check that with the announced values for $\tau, \gamma, a$, we have $\frac{a+1}{a}\tau + 107=1211$ and $\tau + \gamma + 1=1216$.

\subsection{Encoding \textsc{Mim/Sim-Balancing} in \textsc{Mim/Sim-Width}}

We start with the description of a~gadget for each vertex of~$H$.

\medskip

\textbf{Construction of $\bm{\gad(u)}$.}
For each vertex $u \in V(H)$, the \emph{gadget of $u$}, denoted by $\gad(u)$, is a~graph spanned by a~path $Q_u$ of length $2b \cdot |S(u)|$ made by \emph{concatenating} $b$ copies of a~path~$P_u$.
The path $P_u$ is built as follows.
Recall that in the graph $G$, the set $S(u)$ partitions into $I(u, v_1) \uplus \dots \uplus I(u, v_k)$ where $\{v_1, \dots, v_k\} = N_H(u)$.
Since all weights are multiples of~$a$, $|I(u, v)|$ is a multiple of $a$ for any edge $uv \in E(H)$.
Hence we can write each $I(u, v)$ as a~disjoint union $I(u, v, 1) \uplus \dots \uplus I(u, v, a)$ where each $I(u, v, i)$ has size $\frac{|I(u, v)|}{a}$.

We construct the path $L_i$ whose vertex set is $I(u, v_1, i) \cup I(u, v_2, i) \cup \dots \cup I(u, v_k, i)$, and whose vertices occur in this order along $L_i$.
We define $L$ as the concatenation $L_1 L_2\dots L_a$, i.e., the last vertex of~$L_i$ is made adjacent to the first vertex of~$L_{i+1}$, for every $i \in [a-1]$.
The path $P_u$ is obtained from the 1-subdivision of $L$ by adding a~vertex adjacent to the last vertex of $L_a$; see~\cref{fig:Pu}.

\begin{figure}[h!]
  \centering
  \begin{tikzpicture}[
	vertex/.style={draw, circle, minimum size=3pt, inner sep=0pt},
	subd/.style={draw, circle, minimum size=2pt, fill, inner sep=0pt},
	matching/.style={blue, very thick},
	dummy/.style={very thick},
	]
	\def\s{0.35}
	\def\h{1}
	\def\n{12}
	\def\a{3}
	
	\pgfmathtruncatemacro\p{\n * \a}
	\foreach \i in {1,...,\p}{
		\node[vertex] (a\i) at (\i * \s, 0) {} ;
		\node[subd] (b\i) at (\i * \s + \s / 2, 0) {} ;
		\draw (a\i) -- (b\i) ;
	}
	
	\pgfmathtruncatemacro\pm{\p - 1}
	\foreach \i [count = \ip from 2] in {1,...,\pm}{
		\draw (b\i) -- (a\ip) ;
	}
	
	\foreach \z in {1,...,\a}{
		\foreach \i/\j/\k/\l/\c in {1/4/1/v/yellow,5/7/2/v_1/orange,8/9/3/v_2/red,10/12/4/v_3/purple}{
			\pgfmathtruncatemacro\ii{(\z - 1) * \n + \i}
			\pgfmathtruncatemacro\jj{(\z - 1) * \n + \j}
			\node[draw, rounded corners, inner sep=1.75pt, fit=(a\ii) (a\jj), fill opacity=0.2, fill=\c] (I\k-\z) {} ;
		}
		\pgfmathsetmacro\yf{(\z - 1) * \n + 1}
		\pgfmathsetmacro\yl{(\z - 1) * \n + \n + 0.5}
		\draw[stealth-stealth] (\yf * \s - 0.1, -0.25) to node[midway,below] {$L_\z$} (\yl * \s - 0.05, -0.25) ;
		
		\node at (2.5 * \s + \z * \n * \s - \n * \s, 0.4) {$I(u,v_1,\z)$} ;
	}
	
\end{tikzpicture}
  \caption{The path $P_u$ for a~vertex $u$ with four neighbors $v_1, v_2, v_3, v_4$, and $a = 3$.
    The sizes of $I(u,v_1)$, $I(u,v_2)$, $I(u,v_3)$, $I(u,v_4)$ are 12, 9, 6, 9, respectively; all divisible by~$a$.
    The labels $I(u,v_1,\bullet)$ and $L_{\bullet}$ refer to the white vertices, while $P_u$ also comprises the subdivision vertices in black.}
  \label{fig:Pu}
\end{figure}
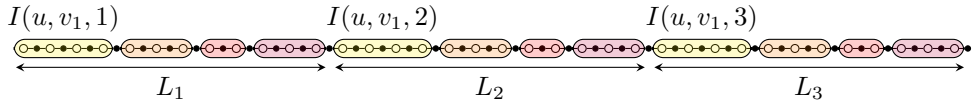

\def\copies{\mathsf{Copies}}

We obtain the path $Q_u$ by concatenating $b$ copies  $P_u^1,P_u^2, \dots, P_u^b$ of $P_u$. 
Note that each vertex $x$ of $P_u$ has $b$ copies $x_1,\dots,x_b$ in $Q_u$; for each $y\in \{x,x_1,\dots,x_b\}$, we denote by $\copies(y)$ the set $\{x_1,\dots,x_b\}$.
The gadget $\gad(u)$ is obtained from $Q_u$ by adding an edge between every pair of vertices $x, y$ in two distinct $P_u^i, P_u^j$ except if $y$ is in $N_{Q_u}[\copies(x)]$; see~\cref{fig:Qu}.

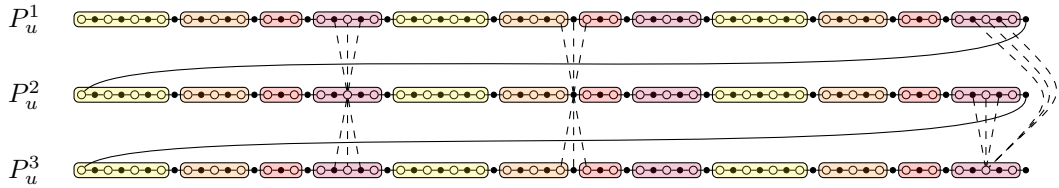
\begin{figure}[h!]
  \centering
  \begin{tikzpicture}[
	vertex/.style={draw, circle, minimum size=3pt, inner sep=0pt},
	subd/.style={draw, circle, minimum size=2pt, fill, inner sep=0pt},
	matching/.style={blue, very thick},
	dummy/.style={very thick},
	]
	\def\sc{1}
	\def\s{0.35}
	\def\h{1}
	\def\n{12}
	\def\a{3}
	\def\b{3}
	
	\foreach \r in {1,...,\b}{
		\pgfmathsetmacro\xs{1.58 * \b * \r / \sc}
		
			\pgfmathtruncatemacro\p{\n * \a}
			\foreach \i in {1,...,\p}{
				\node[vertex] (a\r-\i) at (\i * \s, -\r +1) {} ;
				\node[subd] (b\r-\i) at (\i * \s + \s / 2, -\r +1) {} ;
				\draw (a\r-\i) -- (b\r-\i) ;
			}
			
			\pgfmathtruncatemacro\pm{\p - 1}
			\foreach \i [count = \ip from 2] in {1,...,\pm}{
				\draw (b\r-\i) -- (a\r-\ip) ;
			}
			
			\foreach \z in {1,...,\a}{
				\foreach \i/\j/\k/\l/\c in {1/4/1/v/yellow,5/7/2/v_1/orange,8/9/3/v_2/red,10/12/4/v_3/purple}{
					\pgfmathtruncatemacro\ii{(\z - 1) * \n + \i}
					\pgfmathtruncatemacro\jj{(\z - 1) * \n + \j}
					\node[draw, rounded corners=1.8pt, inner sep=1.1pt, fit=(a\r-\ii) (a\r-\jj), fill opacity=0.2, fill=\c] (I\r-\k-\z) {} ;
				}
			}
			
				\node at (-0.4, -\r +1) {$P_u^\r$} ;

	}
	
	\draw (b1-36) to[out=-90,in=50, looseness = 0.22] (a2-1) ;
	\draw (b2-36) to[out=-90,in=45, looseness = 0.2] (a3-1) ;
	
	\foreach \i in {b1-10, a1-11, b1-11, b3-10, a3-11, b3-11}{
		\draw[dashed, black] (a2-11) to (\i) ;
	}

	\foreach \i in {a1-19, b1-19, a1-20, a3-19, b3-19, a3-20}{
		\draw[dashed, black] (\i) to  (b2-19) ;
	}
	
	\foreach \i in {b2-34, a2-35, b2-35}{
		\draw[dashed, black] (\i) to (a3-35);
	}
	\foreach \i in {b1-34, a1-35, b1-35}{
		\draw[dashed, black]  (\i) to[out=-45,in=45, looseness = 2] (a3-35) ;
	}
\end{tikzpicture}
  \caption{The gadget $\gad(u)$ for the path $P_u$ of~\cref{fig:Pu} and $b = 3$.
    We drew the non-edges between distinct copies of $P_u$ (dashed edges) incident to only three vertices (one subdivision vertex in $P_u^2$, and two regular vertices in $P_u^2$ and $P_u^3$).
  Each path $P_u^i$ remains induced in~$\gad(u)$.}
  \label{fig:Qu}
\end{figure}


\medskip

\textbf{Construction of $\bm{G^*}$.}
Finally, we construct $G^*$ as follows (based on the vertex set of~$H$, and the edge set of~$G$).
For each vertex $u \in V(H)$, we add a~gadget $\gad(u)$ to $G^*$.  
For every edge $xy \in E(G)$, we add the biclique between $\copies(x)$ and $\copies(y)$ in $G^*$.
If $xy$ is a~matching edge of $G$, the added edges are also said \emph{matching} (see~\cref{fig:gad(u)-adjacencies-matching}).
Similarly if $xy$ is a~dummy edge, we call the added edges \emph{dummy} (see~\cref{fig:gad(u)-adjacencies-dummy}).

\begin{figure}[h!]
  \centering
  \begin{tikzpicture}[
	vertex/.style={draw, circle, minimum size=3pt, inner sep=0pt},
	subd/.style={draw, circle, minimum size=2pt, fill, inner sep=0pt},
	matching/.style={blue, very thick},
	dummy/.style={very thick},
	]
	\def\sc{0.3}
	\def\s{0.4}
	\def\h{1}
	\def\n{12}
	\def\a{3}
	\def\b{3}
	\def\vs{4}
	
	\foreach \r in {1,...,\b}{
		\pgfmathsetmacro\xs{1.58 * \b * \r / \sc}
		
		\begin{scope}[scale=\sc, transform shape, xshift= \xs cm]
			\pgfmathtruncatemacro\p{\n * \a}
			\foreach \i in {1,...,\p}{
				\node[vertex] (a\r-\i) at (\i * \s, 0) {} ;
				\node[subd] (b\r-\i) at (\i * \s + \s / 2, 0) {} ;
				\draw (a\r-\i) -- (b\r-\i) ;
			}
			
			\pgfmathtruncatemacro\pm{\p - 1}
			\foreach \i [count = \ip from 2] in {1,...,\pm}{
				\draw (b\r-\i) -- (a\r-\ip) ;
			}
			
			\foreach \z in {1,...,\a}{
				\foreach \i/\j/\k/\l/\c in {1/4/1/v/yellow,5/7/2/v_1/orange,8/9/3/v_2/red,10/12/4/v_3/purple}{
					\pgfmathtruncatemacro\ii{(\z - 1) * \n + \i}
					\pgfmathtruncatemacro\jj{(\z - 1) * \n + \j}
					\node[draw, rounded corners=1.8pt, inner sep=1.1pt, fit=(a\r-\ii) (a\r-\jj), fill opacity=0.2, fill=\c] (I\r-\k-\z) {} ;
				}
			}
		\end{scope}
		
		\node at (\xs * \sc + 0.48 * \n * \s, -0.4) {$P_u^\r$} ;
	}
	
	\draw (b1-36) -- (a2-1) ;
	\draw (b2-36) -- (a3-1) ;
	
	\foreach \r in {1,...,\b}{
		\pgfmathsetmacro\xs{1.58 * \b * \r / \sc}
		
		\begin{scope}[scale=\sc, transform shape, xshift= \xs cm]
			\pgfmathtruncatemacro\p{\n * \a}
			\foreach \i in {1,...,\p}{
				\node[vertex] (c\r-\i) at (\i * \s, \vs / \sc) {} ;
				\node[subd] (d\r-\i) at (\i * \s + \s / 2, \vs / \sc) {} ;
				\draw (c\r-\i) -- (d\r-\i) ;
			}
			
			\pgfmathtruncatemacro\pm{\p - 1}
			\foreach \i [count = \ip from 2] in {1,...,\pm}{
				\draw (d\r-\i) -- (c\r-\ip) ;
			}
			
			\foreach \z in {1,...,\a}{
				\foreach \i/\j/\k/\l/\c in {1/3/1/v/green!95!blue,4/7/2/v_1/green!65!blue,8/10/3/v_2/green!35!blue,11/12/4/v_3/green!5!blue}{
					\pgfmathtruncatemacro\ii{(\z - 1) * \n + \i}
					\pgfmathtruncatemacro\jj{(\z - 1) * \n + \j}
					\node[draw, rounded corners=1.8pt, inner sep=1.1pt, fit=(c\r-\ii) (c\r-\jj), fill opacity=0.2, fill=\c] (J\r-\k-\z) {} ;
				}
			}
		\end{scope}
		
		\node at (\xs * \sc + 0.48 * \n * \s, \vs + 0.4) {$P_v^\r$} ;
	}
	
	\draw (d1-36) -- (c2-1) ;
	\draw (d2-36) -- (c3-1) ;

	\foreach \y in {1,...,\b}{
		\foreach \z in {1,...,\b}{
			\foreach \i in {1,...,4,13,14,...,16,25,26,...,28}{
				\pgfmathtruncatemacro\j{\i+3}
				\draw[thin, blue] (a\y-\i) -- (c\z-\j) ;
			}
		}
	}
\end{tikzpicture}
  \caption{The matching edges between $\gad(u)$ and $\gad(v)$ (with $u$ and $v$ two adjacent vertices in~$H$).}
  \label{fig:gad(u)-adjacencies-matching}
\end{figure}
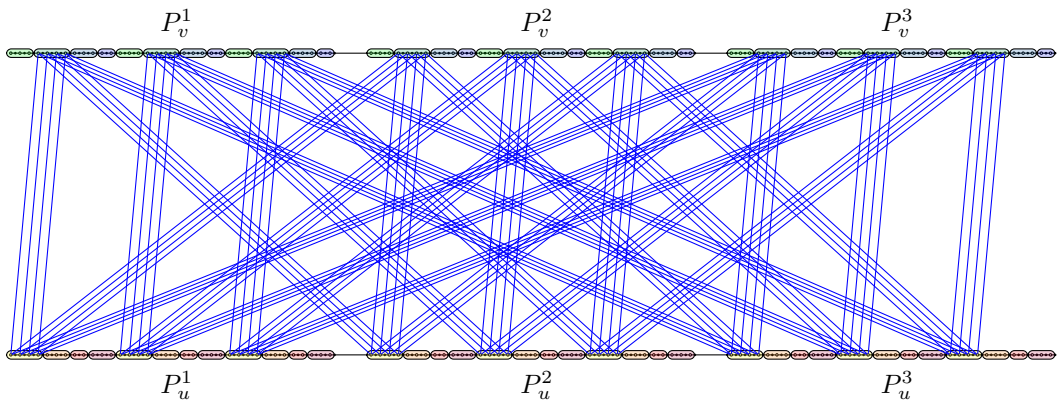
\begin{figure}[h!]
  \centering
  \begin{tikzpicture}[
	vertex/.style={draw, circle, minimum size=3pt, inner sep=0pt},
	subd/.style={draw, circle, minimum size=2pt, fill, inner sep=0pt},
	matching/.style={blue, very thick},
	dummy/.style={very thick},
	]
	\def\sc{0.3}
	\def\s{0.4}
	\def\h{1}
	\def\n{12}
	\def\a{3}
	\def\b{3}
	\def\vs{4}
	
	\foreach \r in {1,...,\b}{
		\pgfmathsetmacro\xs{1.58 * \b * \r / \sc}
		
		\begin{scope}[scale=\sc, transform shape, xshift= \xs cm]
			\pgfmathtruncatemacro\p{\n * \a}
			\foreach \i in {1,...,\p}{
				\node[vertex] (a\r-\i) at (\i * \s, 0) {} ;
				\node[subd] (b\r-\i) at (\i * \s + \s / 2, 0) {} ;
				\draw (a\r-\i) -- (b\r-\i) ;
			}
			
			\pgfmathtruncatemacro\pm{\p - 1}
			\foreach \i [count = \ip from 2] in {1,...,\pm}{
				\draw (b\r-\i) -- (a\r-\ip) ;
			}
			
			\foreach \z in {1,...,\a}{
				\foreach \i/\j/\k/\l/\c in {1/4/1/v/yellow,5/7/2/v_1/orange,8/9/3/v_2/red,10/12/4/v_3/purple}{
					\pgfmathtruncatemacro\ii{(\z - 1) * \n + \i}
					\pgfmathtruncatemacro\jj{(\z - 1) * \n + \j}
					\node[draw, rounded corners=1.8pt, inner sep=1.1pt, fit=(a\r-\ii) (a\r-\jj), fill opacity=0.2, fill=\c] (I\r-\k-\z) {} ;
				}
			}
		\end{scope}
		
		\node at (\xs * \sc + 0.48 * \n * \s, -0.4) {$P_u^\r$} ;
	}
	
	\draw (b1-36) -- (a2-1) ;
	\draw (b2-36) -- (a3-1) ;
	
	\foreach \r in {1,...,\b}{
		\pgfmathsetmacro\xs{1.58 * \b * \r / \sc}
		
		\begin{scope}[scale=\sc, transform shape, xshift= \xs cm]
			\pgfmathtruncatemacro\p{\n * \a}
			\foreach \i in {1,...,\p}{
				\node[vertex] (c\r-\i) at (\i * \s, \vs / \sc) {} ;
				\node[subd] (d\r-\i) at (\i * \s + \s / 2, \vs / \sc) {} ;
				\draw (c\r-\i) -- (d\r-\i) ;
			}
			
			\pgfmathtruncatemacro\pm{\p - 1}
			\foreach \i [count = \ip from 2] in {1,...,\pm}{
				\draw (d\r-\i) -- (c\r-\ip) ;
			}
			
			\foreach \z in {1,...,\a}{
				\foreach \i/\j/\k/\l/\c in {1/3/1/v/green!95!blue,4/7/2/v_1/green!65!blue,8/10/3/v_2/green!35!blue,11/12/4/v_3/green!5!blue}{
					\pgfmathtruncatemacro\ii{(\z - 1) * \n + \i}
					\pgfmathtruncatemacro\jj{(\z - 1) * \n + \j}
					\node[draw, rounded corners=1.8pt, inner sep=1.1pt, fit=(c\r-\ii) (c\r-\jj), fill opacity=0.2, fill=\c] (J\r-\k-\z) {} ;
				}
			}
		\end{scope}
		
		\node at (\xs * \sc + 0.48 * \n * \s, \vs + 0.4) {$P_v^\r$} ;
	}
	
	\draw (d1-36) -- (c2-1) ;
	\draw (d2-36) -- (c3-1) ;

	\foreach \x in {1,...,\b}{
		\foreach \y in {1,...,\a}{
			\foreach \z in {1,...,\a}{
				\foreach \i/\j in {2/1,2/3,2/4, 3/1,3/3,3/4, 4/1,4/3,4/4}{
					\draw[thin] (I2-\i-\y) -- (J\x-\j-\z) ;
				}
			}
		}
	}
\end{tikzpicture}
  \caption{The dummy edges between $\gad(u)$ and $\gad(v)$ of~\cref{fig:gad(u)-adjacencies-matching}.
  An edge between two rounded boxes represents a~biclique between the corresponding vertices of $I(\bullet,\bullet,\bullet)$ (which excludes the subdivision vertices).
  We only represented the bicliques incident to $P_u^2$ ($P_u^1$ and $P_u^3$ have the same adjacencies toward~$\gad(v)$). }
  \label{fig:gad(u)-adjacencies-dummy}
\end{figure}
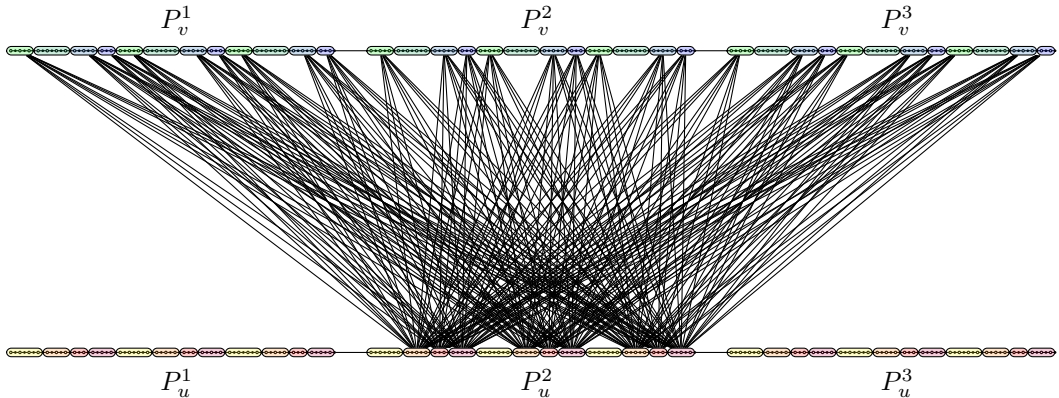

\subsection{Low linear mim-balancing of~$(G,\mathcal P)$ $\Rightarrow$ low linear-mim width of~$G^*$}\label{sec:lin-mim-balancing-to-lin-mim-width}

To upper bound the linear mim-width of $G^*$, we need the next lemma on~$\gad(u)$.
For each vertex $u$ of $H$, we define the \emph{caterpillar layout} of~$\gad(u)$ as the \emph{left-aligned} caterpillar (i.e, such that every right child is a~leaf) layout with $|V(\gad(u))|$ leaves bijectively labeled by $V(\gad(u))$, in the order of~$Q_u$ from the first vertex of $P_u^1$ to the last vertex of $P_u^b$.

\begin{lemma}\label{lem:mim-value:gadget}
Let $u$ be a vertex of $H$, and $(C,f)$ be the caterpillar layout of $\gad(u)$. 
For every cut $(A,B)$ of $\gad(u)$ induced by an edge of $C$, the mim-value of $(A,B)$ is at most 7.
\end{lemma}
\begin{proof}
	Let $(A,B)$ be a cut induced by an edge of $C$.
	Since the leaves of $C$ are bijectively mapped $f$ to $V(\gad(u))$ in the order of~$Q_u$, there is exactly one edge $e_{Q_u}$ between $A$ and $B$ that belongs to $Q_u$.
	Let $P_u^i$ be the copy of $P_u$ such that $e_{Q_u}$ is either an edge of $P_u^i$ or the edge between $P_u^{i-1}$ and $P_u^{i}$.
	
	Let $M$ be an induced matching of $\gad(u)[A,B]$ of size at~least 2.
	Let $X$ and $Y$ be the endpoints of the edges of $M$ lying in $A$ and $B$, respectively.
	Observe that there exists at least one vertex in $X\cup Y$ that is neither on $P_u^i$ nor an endpoint of $e_{Q_u}$.
	Indeed, we have two cases to consider:
	\begin{compactitem}
		\item Case 1: $e_{Q_u}$ is the edge between $P_u^{i-1}$ and $P_u^{i}$. Then, $V(P_u^i)$ is fully included in either $A$ or $B$, and since the sizes of $X$ and $Y$ are at least 2, it follows that at least one vertex in $X\cup Y$ is neither on $P_u^i$ nor an endpoint of $e_{Q_u}$.
		\item Case 2: $e_{Q_u}$ is an edge of $P_u^i$. Then, since $M$ contains at~least two edges and $e_{Q_u}$ is the only edge of $P_u^i$ between $A$ and $B$, at least one vertex in $X\cup Y$ is not on $P_u^i$ (and thus not an endpoint of $e_{Q_u}$).
	\end{compactitem}
	We assume, without loss of generality, that $X$ contains a~vertex $x$ that is neither on $P_u^i$ nor an endpoint of $e_{Q_u}$ (we can always swap $X$ and $Y$).
	In the following, we prove that $x$ has at most 6 non-neighbors in $Y$.
        This is sufficient to prove the lemma as it implies that the size of~$Y$, and thus of~$M$, is at most 7; as desired.
	
	Let $P^j_u$ be the copy of $P_u$ containing $x$.
	Observe that we have $i\neq j$ and thus all the vertices of $P_u^j$ belong to $A$.
	We denote by $x^-$ and $x^+$ the neighbors of $x$ in $P_u^j$ (we may have $x^-=x^+$ when $x$ is an endpoint of $P_u^j$).
	Since all the vertices of $P_u^j$ belong to $A$ and $x$ is not incident to $e_{Q_u}$, we have $B \setminus N_{\gad(u)}(x) = B\cap (\copies(x)\cup \copies(x^-)\cup \copies(x^+))$. 
	Since $P_u^i$ contains exactly one copy of each vertex among $x$, $x^-$ and $x^+$, it follows that $x$ has at most 3 non-neighbors in $B\cap V(P_u^i)$ and in particular in $Y\cap V(P_u^i)$.
	It remains to prove that $Y\setminus V(P_u^i)$ contains at most 3 non-neighbors of $x$.
	Observe that for each $y\in \{x,x^-, x^+\}$ and every pair of vertices $w,z$ in $(B\cap \copies(y)) \setminus V(P_u^i)$, we have $N_{\gad(u)}(w) \cap A = N_{\gad(u)}(z)\cap A$ and thus at most one vertex among $w$ and $z$ can be in $Y$.
	We conclude that $x$ has at most 3 non-neighbors in $Y\setminus V(P_u^i)$ and thus $|Y|=|M| \leq 7$.	
\end{proof}

We can now conclude for this direction of the reduction.

\begin{lemma}
If $(H, \omega)$ admits a~$\tau$-balancing order, then the linear mim-width of $G^*$ is at~most $\frac{a+1}{a} \tau + 107$.
\end{lemma}
\begin{proof}
	Suppose that $(H, \omega)$ admits a $\tau$-balancing order $\prec$.
	By \Cref{lem:balanced2sim-width}, $G$ admits a~\pmap~$(P, f)$ of mim-value at most $\tau + 50$.
	We construct from $\prec$ a caterpillar layout $C$ of $G^*$, and leverage the mim-value of $(P, f)$ to show that the mim-value of~$C$ is at most $\frac{a+1}{a} \tau + 107$.
	
	For every $i \in [|V(H)|]$, we denote by $u_i$ the $i$-th vertex of $H$ along~$\prec$.
	We denote by~$C_i$ the caterpillar layout of $\gad(u_i)$, and we denote by $S_i$ the spine of~$C_i$.	
	We construct a caterpillar layout $C$ of $G^*$ from the disjoint union of $C_1,\dots,C_{|V(H)|}$ by adding an edge between the last vertex of $S_i$ and the first vertex of $S_{i+1}$ for each $i\in [|V(H)-1]$.
	Let $\mathcal S^* :=\{ V(\gad(u))~:~u\in V(H)\}$.
	Let us recall that an $\mathcal S^*$-cut $(A,B)$ is a cut of $G^*$ such that no gadget $\gad(u)$ has vertices in both $A$ and $B$. 
	Observe that a cut induced by an edge $e$ of $C$ is an $\mathcal S^*$-cut if and only if $e$ is an edge between two caterpillars $C_i$ and $C_{i+1}$.
	
	\begin{claim}\label{claim:glue-edges}
		Every $\mathcal S^*$-cut $(A,B)$ induced by an edge of $C$ has mim-value at most $\tau + 50$.
		Moreover, every semi-induced matching $M$ of $G^*$ between $A$ and $B$ has at most $\tau$ matching edges, and there exists $u\in V(H)$ such that $V(\gad(u))$ covers the matching edges of $M$.
	\end{claim}
	\begin{proofofclaim}		
		Let $(A,B)$ be a $\mathcal S^*$-cut induced by an edge $e$ of $C$.
		Let $i\in [|V(H)|-1]$ such that $e$ is the edge between $C_i$ and $C_{i+1}$.
		We denote by $e_P$ the $i$-th edge of $P$ (recall that $(P,f)$ is a~\pmap of~$G$) and by $(A_P,B_P) := (A_{e_P}, B_{e_P})$ the cut of $G$ induced by $e_P$.
		Let $M$ be a semi-induced matching of $G^*$ between $A$ and $B$.
		
		We claim that $|M| \leq \mim_{G}(A_P,B_P)$.
		Observe that $A_P$ is the union of $S(u_1),\dots,S(u_i)$ and $B_P$ is the union of $S(u_{i+1}),\dots,S(u_{|V(H)|})$.
		Similarly, $A$ is the union of $\gad(u_1),\dots,\gad(u_i)$ and $B$ is the union of $\gad(u_{i+1}),\dots,\gad(u_{|V(H)|})$.
		For every vertex $v$ of $G$, we say that $v$ is the \emph{original} of the vertices of $G^*$ in $\copies(v)$.
		
		Since $(A,B)$ is a $\mathcal S^*$-cut, every edge in $G^*[A,B]$ is between two different gadgets of $G^*$.
		By construction of $G^*$, for every edge $xy$ between two gadgets $\gad(u_i)$ and $\gad(u_j)$, the vertices $x$ and $y$ are the copies of some vertices in $G$. 
		Hence, every endpoint of an edge in~$M$ has an original in~$G$.
		
		By construction of~$G^*$, for all vertices $x$ and $y$ in~$G^*$ with distinct originals $w$ and $z$ in $G$, we have $xy\in E(G^*)$ if and only $wz\in E(G)$.
		Hence, by replacing every edge $xy$ in $M$ by a~pair $\{w,z\}$ where $w$ and $z$ are the originals of $x$ and $y$, respectively, we obtain a~semi-induced matching $|M_P|$ between $A_P$ and $B_P$.
		Hence, we have $ |M| \leq |M_P| \leq \mim_{G}(A_P,B_P)$.	
		As the mim-value of $(P,f)$	is at most $\tau + 50$, we conclude that the size of $M$ is at most $\tau + 50$.
		
		By \Cref{lem:balanced2sim-width}, we know that $M_P$ has at most $\tau$ matching edges and that there exists $u\in V(H)$ such that $S(u)$ covers the matching edges of $M_P$.
		As the copies of the vertices in $S(u)$ are all in $\gad(u)$, the vertices in $\gad(u)$ cover the matching edges of $M$.
	\end{proofofclaim}
	
	Next we deal with the cuts of~$C$ that are not $\mathcal S^*$-cuts.
	Let $(A,B)$ be a cut induced by an edge~$e$ of~$C$ that is not a $\mathcal S^*$-cut.
	Then, there exists $k\in [|V(H)|]$ such that $e$ is the edge of some caterpillar layout $C_k$.
	If a~leaf of~$C_k$ is incident to $e$, then $A$ or $B$ is a singleton (containing exactly one vertex in $\gad(u_k)$) and $\mim_{G^*}(A,B)\leq 1$.
	In the remainder of the proof, we assume that $e$ is an edge from the spine of~$C_k$.
	
	Let $M=\{e_1,\dots,e_m\}$ be a semi-induced matching between $A$ and $B$.
	By \Cref{lem:mim-value:gadget}, we can assume that $M':=\{e_1,\dots,e_{m-7}\}$ contains no edge within $\gad(u_k)$.
	We denote by $A_k$ and $B_k$ the sets of vertices of $\gad(u_k)$ that are respectively in $A$ and $B$.
	Let $M_A$ (resp.~$M_B$) be the sets of edges in $M$ between $A$ and $B \setminus B_k$ (resp.~between $B$ and $A \setminus A_k$).
	Observe that the edges of $M_A$ are traversing the cut $(A\cup B_k, B\setminus B_k)$ which is a $\mathcal S^*$-cut induced by an edge of~$C$.
	Symmetrically, the edges of $M_B$ are also traversing a $\mathcal S^*$-cut induced by an edge of~$C$.
	From \Cref{claim:glue-edges}, for each $X\in \{A,B\}$, we know that $M_X$ has at most $\tau + 50$ edges and at most $\tau$ matching edges, moreover there exists $v_X\in V(H)$ such that $V(\gad(v_X))$ covers the matching edges of $M_X$.
	Let $\widehat M,\widehat M_A$ and $\widehat M_B$ be the sets of matching edges from respectively $M',M_A$ and $M_B$.
	Since $\widehat M = \widehat M_A \cup \widehat M_B$, and we remove at most 50 edges from $M_A$ and $M_B$ to obtain respectively $\widehat M_A$ and $\widehat M_B$, we have 
	\begin{equation}\label{eq:M:hat:M}
		|M| \leq |\widehat M| + 107.
	\end{equation}
	
	\begin{claim}
		There exists $v \in \{v_A,v_B,u_k\}$ such that $V(\gad(v))$ covers $\widehat M$.
	\end{claim}
	\begin{proofofclaim}
		Assume toward a~contradiction that the claim is false.
		As $V(\gad(v_A))$ covers $\widehat M_A$ but not $\widehat M$, it means that there exists at least one edge $e_A$ in $\widehat M\setminus \widehat M_A$ not incident to $V(\gad(v_A))$.
		As $e_A \notin \widehat M_A$, this edge must be in $\widehat M_B$ between $A$ and $B_k$, so it must be covered by $V(\gad(v_B))$. 
		Since $B_k\subseteq V(\gad(u_k))$, $e_A$ is a matching edge between $\gad(v_B)$ and $\gad(u_k)$.
		Symmetrically, we deduce that there is a matching edge between $\gad(v_A)$ and $\gad(u_k)$ (because $V(\gad(v_B))$ covers $\widehat M_B$ but not $\widehat M$).
		Moreover, as $V(\gad(u_k))$ does not cover $\widehat M$, there exists an edge $e_{u_k}$ in $\widehat M$ not adjacent to $\gad(u_k)$.
		As $A_k$ and $B_k$ are subsets of $V(\gad(u_k))$, $e_{u_k}$ must be in $\widehat M_A$ and $\widehat M_B$, so it must be covered by both $V(\gad(v_A))$ and $V(\gad(v_B))$.
		So there is at least one matching edge between every pair of gadgets among $\gad(v_A), \gad(v_B)$ and $\gad(u_k)$.
		From the construction of $G^*$, it means that there is at least one matching edge between $S(v_A)$, $S(v_B)$ and $S(u_k)$.
		Recall that a matching edge in $G$ between two parts of $S(u)$, $S(v)$ in $\mathcal S$ implies that $uv$ is an edge of $H$.
		Consequently, $v_A, v_B$ and $u_k$ induce a triangle in $H$, a contradiction with $H$ being triangle-free.
	\end{proofofclaim}
	
	First, suppose that $v_A\neq u_k$ and that $V(\gad(v_A))$ covers $\widehat M$.
	As $v_A\neq u_k$, there is no edge in $\widehat M$ between $A_k$ and $B$ (such edges would not be covered by $V(\gad(v_A))$).
	So, we have $\widehat M = \widehat M_A$ and since $|\widehat M_A| \leq \tau$, it follows by \Cref{eq:M:hat:M} that $|M| \leq \tau + 107$.
	By symmetry, the above holds also when $v_B\neq u_k$ and  $V(\gad(v_B))$ covers $\widehat M$.
	
	Now, we assume that $V(\gad(u_k))$ covers $\widehat M$.
	We distinguish two cases.
	\begin{itemize}
		\item Case 1: There exists $\ell \in [b]$ such that $V(P_{u_k}^\ell) \subseteq A$ and at least one endpoint $x$ of~$\widehat M$ is on~$P_{u_k}^\ell$. 
		We claim that $\widehat M$ has at most 3 edges with endpoints in~$B_k$.
		Recall that $\widehat M$ has only matching edges, so every edge of $\widehat M$ is between two distinct gadgets.
                In particular, the endpoints of $\widehat M$ in $B_k$ must be non-neighbors of $x$ and adjacent via $\widehat M$ to a vertex in $A\setminus A_k$.
		Since $V(P_{u_k}^\ell) \subseteq A$, the non-neighborhood of $x$ in $B_k$ is exactly $B_k \cap (\copies(x) \cup \copies(x^-) \cup \copies(x^+))$, where $x^-$ and $x^+$ are the neighbors of $x$ in $P_{u_k}^\ell$ (possibly with $x^-=x^+$).
		Thus, the endpoints of $\widehat M$ in $B_k$ must be in $\copies(x) \cup \copies(x^-) \cup \copies(x^+)$.
		However, for each $y\in \{x,x^-,x^+\}$, the vertices in $B_k \cap \copies(y)$ have the same neighborhood in $A\setminus A_k$.
		We deduce that $\widehat M$ has at most $3$ endpoints in $B_k$.
		By removing the edges of $\widehat M$ with an endpoint in $B_k$, we obtain $\widehat M_A$ which contains at most $\tau$ edges.
		Hence, we have $|\widehat M| \leq \tau + 3$ and by \Cref{eq:M:hat:M}, it follows that $|M| \leq \tau + 110$.
		
		\item Case 2: There exists $\ell\in [b]$  such that $P_{u_k}^\ell$ has vertices in both $A$ and $B$, and every endpoint of $\widehat M$ in $V(\gad(u_k))$ is from $P_{u_k}^\ell$.
		Here, we have to use the balanced distribution of the copies of vertices from $S(u_k)$ along the path $P_{u_k}^\ell$. 
		Recall that $\widehat M$ contains only matching edges with one endpoint on $P_{u_k}^\ell$. Thus, for every edge $xy$ of $\widehat M$ with $x$ on $P_{u_k}^\ell$, there exists $v\in N_H(u_k)$ such that $y$ is from $\gad(v)$; in particular, $x$ and $y$ are the copies of vertices in $I(u_k,v)$ and $I(v,u_k)$, respectively.
		In this setting, we call $xy$ a $u_kv$-edge. 
		For each $X\in \{A,B\}$, let $N_H(u_k)^X$ be set of neighbors $v$ of $u_k$ in $H$ such that all the vertices of $\gad(v)$ are in $X$.
		Notice that each edge in $\widehat M_X$ is an $u_kv$-edge with $v\in N_H(u_k)^X$.
		
		For every $u_kv$-edge $xy$ of $\widehat M$ where $y$ is from $\gad(v)$, we assume without loss of generality that the vertex $y$ is from $P_{v}^1$.
		This assumption can be made since the vertices in $\copies(y)$ have the same neighborhood in $G^*[A,B]$, so we can always replace $y$ in $\widehat M$ with its copy in the path $P^1_v$.
		For each $v\in N_H(u_k)$, let $M_{u_kv}$ be the induced matching of size~$\omega(u_kv)$ made of the matching edges between $P_{u_k}^{\ell}$ and $P_{v}^{1}$.
		For each $X\in \{A,B\}$, we define $M_X^\star$ as the union of the matchings $M_{u_kv}$ over the vertices $v\in N_H(u_k)^{X}$.
		By our previous assumption on $\widehat M$, we have $\widehat M_A \subseteq M_A^\star$ and $\widehat M_B \subseteq M_B^\star$.
		As $M_A^\star$ contains only matching edges that traverse the $\mathcal S^*$-cut $(A\setminus A_k, B\cup A_k)$, we know that $|M_A^\star| \leq \tau$ by~\Cref{claim:glue-edges}.
		Symmetrically, We have $|M_B^\star| \leq \tau$.
		
		Recall that from the construction of $P_{u_k}$, the path  $P_{u_k}^\ell$ is the concatenation of $a$ paths $P_1',\dots,P_a'$ such that for each $i\in [a]$ and $v\in N_H(u_k)$, $P_i'$ contains $\frac{\omega(u_kv)}{a}$ endpoints of $M_{u_kv}$.
		Let $t\in [a]$ be such that the first vertex of $P_{u_k}^\ell$ in~$B$ is from $P_t'$.
		Observe that all the vertices from $P_1',\dots,P_{t-1}'$ belong to~$A$ and all the vertices from $P_{t+1}',\dots,P_{a}'$ belong to~$B$.
		Moreover, $P_t'$ is the only path among $P_1',\dots,P_a'$ that can have vertices in both~$A$ and~$B$.
		Consequently, for every $v\in N_H(u_k)$, the endpoints of $M_{u_kv}$ in $A_k$ lie in $P_1',\dots,P_t'$ and those in $B_k$ lie in $P_t',\dots,P_a'$.
		So, $M_{u_kv}$ has at most $\omega(u_kv) t/a$ endpoints in $A_k$ and $\omega(u_kv)(a-t+1)/a$ endpoints in $B_k$.
		We deduce that:
		\begin{itemize}
			\item $M_A$ has at most $(a-t+1)/a$ endpoints in $B_k$,
			\item $M_B$ has at most $t/a$ endpoints in $A_k$.
		\end{itemize}
		As every edge in $\widehat M_A$ is between $A$ and $B_k$, it follows that $|\widehat M_A | \leq \frac{a-t+1}{a} |M_A^\star|$.
		Symmetrically, we have $|\widehat M_B| \leq \frac{t}{a} |M_B^\star|$.
		Since the sizes of $M_A^\star$ and $M_B^\star$ are at most $\tau$, we have
		\begin{equation*}
			|\widehat M| = |\widehat M_A| + |\widehat M_B| \leq \frac{a-t+1}{a} |M_A^\star| + \frac{t}{a} |M_B^\star|   \leq \frac{a+1}{a} \tau.
		\end{equation*}
		By~\Cref{eq:M:hat:M}, we get that $|M| \leq \frac{a+1}{a} \tau + 107$.
	\end{itemize}
	We conclude  that $M$ has at most $\max(\tau + 110,  \frac{a+1}{a} \tau + 107) =  \frac{a+1}{a} \tau + 107$ edges (since $\tau \geqslant 3a$).
\end{proof}

\subsection{Low sim-width of $G^*$ $\Rightarrow$ low sim-balancing of~$(G,\mathcal P)$}\label{sec:sim-width-to-sim-balancing}

%

The main argument works as follows.
We will prove that in any tree layout of $G^*$ of small sim-value, one can associate to each gadget $\gad(u)$ an edge $e$ of the layout, such that a~copy of $P_u$ can be found on both sides of the cut defined by $e$.
Since $\gad(u)$ is ``represented'' by any of its copies of $P_u$, we will use $e$ as where the vertices of~$\gad(u)$ should be moved to.
With this in mind, we gradually build a~\tmap of $G$ from a~tree layout of $G^*$ by ``relocating'' each gadget to the edge it is associated to.

We start with two technical lemmas.

\begin{lemma}\label{lem:cut-doesnt-cut}
Let $P_1, \dots, P_k$ be $k$ paths with $k > t \cdot \max_{i \in [k]} |V(P_i)|$ such that for any $i \neq j \in [k]$ and any $\ell \in [\min(|V(P_i)|,|V(P_j)|)-1]$, the $\ell$-th edge of $P_i$ is mutually induced with the $\ell$-th edge of $P_j$.
Then for any cut $(A, B)$ of sim-value at most $t$, there exists at least one path $P_i$ such that $V(P_i)\subseteq A$ or $V(P_i)\subseteq B$.
\end{lemma}
\begin{proof}
	Let $(A, B)$ be a cut such that every path $P_i$ intersects both $A$ and $B$.
	We claim that the mim-value of $(A,B)$ is at least $t+1$.
	Observe that for each $i \in [k]$, $P_i$ has at least one edge $e_i$ with one endpoint in $A$ and the other in $B$.
	Since $k > t \cdot \max_{i \in [k]} |V(P_i)|$, from the pigeonhole principle, there exists $i_1,\dots,i_{t+1}$ such that the edges $e_{i_1},\dots,e_{i_{t+1}}$ are the $\ell$-th edges of respectively $P_{i_1}, \ldots,P_{i_{t+1}}$ for some $\ell$.
        Thus, by assumption, $e_{i_1},\dots,e_{i_{t+1}}$ form an induced matching, and the sim-value of $(A,B)$ is at~least $t+1$.
\end{proof}

\begin{lemma}\label{lem:tricut-doesnt-cut}
Let $P_1, \dots, P_k$ be $k$ paths with $k > \left \lceil \frac{3t}{2} \right \rceil \cdot \max_{i \in [k]} |V(P_i)|$ such that for any $i \neq j \in [k]$ and any $\ell \in [\min(|V(P_i)|,|V(P_j)|)-1]$, the $\ell$-th edge of $P_i$ is mutually induced with the $\ell$-th edge of $P_j$.
Then for any tripartition $(A, B, C)$ of $V(P_1) \cup \ldots \cup V(P_k)$ such that the cuts $(A,B \cup C)$, $(B,A \cup C)$ and $(C, A \cup B)$ have sim-value at most $t$, there exists at least one path $P_i$ such that $V(P_i)$ is included in one set among $A$, $B$ and $C$.
\end{lemma}
\begin{proof}
	Let $(A,B,C)$ be a tripartition of $V(P_1) \cup \ldots \cup V(P_k)$ such that every path $P_i$ intersects at~least two sets among $A,B$ and $C$.
	We claim that the sim-value of one cut among $(A,B\cup C)$, $(B,A\cup C)$ and $(C, A\cup B)$ is at least $t+1$.
	Observe that for each $i\in[k]$, $P_i$ has at least one edge $e_i$ whose endpoints lie in different sets among $A,B$ and $C$.
	Let $r := \left \lceil \frac{3t}{2}\right \rceil +1$.
	Since $k > \lceil \frac{3t}{2} \rceil \cdot \max_{i \in [k]} |V(P_i)|$, from the pigeonhole principle, there exists $i_1, \ldots, i_r$ such that the edges $e_{i_1}, \ldots, e_{i_r}$ are the $\ell$-th edges of $P_{i_1}, \ldots, P_{i_{r}}$, respectively, for some $\ell$.
	
	Let $S(A,B), S(A,C)$ and $S(B,C)$ be the sets of edges among $e_{i_1},\dots,e_{i_{r}}$ whose endpoints lie in $A\cup B$, $A\cup C$ and $B\cup C$, respectively.
	As the edges $e_{i_1},\dots,e_{i_{r}}$ form an induced matching, and those in $S(A,B)\cup S(A,C)$ are between $A$ and $B\cup C$, we have $\ssim(A,B\cup C) \geq |S(A,B)|+ |S(A,C)|$.
	Similarly, we have $\ssim(B,A\cup C) \geq |S(A,B)|+ |S(B,C) |$ and $\ssim(C,A\cup B) \geq |S(A,C)|+ |S(B,C) |$.
	It follows that $$\ssim(A,B\cup C) + \ssim(B,A\cup C) + \ssim(C, A\cup B) \geq 2 (|S(A,B)|+  |S(A,C)| + | S(B,C)|).$$
	Since $|S(A,B)|+  |S(A,C)| + |S(B,C)| = r = \lceil \frac{3t}{2} \rceil +1$, we have $$\ssim(A,B\cup C) + \ssim(B,A\cup C) + \ssim(C, A\cup B) \geq 3t +2.$$
	We conclude that the maximum among $\ssim(A,B\cup C), \ssim(B,A\cup C)$ and $\ssim(C, A\cup B)$ is at least $t+1$.
\end{proof}

A \emph{\htree} of a~partitioned graph $(J, \mathcal P)$ is pair $(T, f)$ where $T$ is a tree and $f \colon V(J) \rightarrow V(T)$ is a~map such that for any node $t \in T$, either $f^{-1}(t) \in \mathcal P$ or $|f^{-1}(t)| \le 1$.
As for \tmaps each edge $e \in E(T)$ in a \htree of $(J, \mathcal P)$ defines a cut $(A_e, B_e)$ of $J$: the sets of vertices mapped to each component of $T-e$.
The sim-value of the \htree $(T, f)$ is the maximum sim-value of all possible cuts $(A_e, B_e)$ for $e \in E(T)$.

We recall that $\mathcal S^* = \{V(\gad(u))~:~u \in V(H)\}$.
We observe that in the instances $(H,\weight)$ produced by the first reduction, every vertex has weight at most~$2 \tau$.
Hence $\max_{u \in V(H)} |V(P_u)| \leq 4 \tau$.
We set $\alpha := \lceil \frac{b-1}{6 \tau} \rceil - 1 = \tau + \gamma - 1$, where $b$ is the constant introduced at the beginning of the section.

\begin{lemma}\label{lem:default-edge}
Let $(T, f)$ be a~\htree of $(G^*, \mathcal S^*)$ of sim-value at~most~$\alpha$ and $T$ subcubic.
Then, for any vertex $u$ of $H$, either
\begin{compactitem}
\item there exists $t \in V(T)$ with $f^{-1}(t) = V(\gad(u))$, or
\item there exists an edge $e \in E(T)$ such that in the cut $(A_e, B_e)$ induced by $e$ in $G^*$, one can find a copy of $P_u$ in both $A_e$ and $B_e$.
\end{compactitem}
\end{lemma}

\begin{proof}
Let $u$ be a vertex of $H$.
If there exists a node $t \in V(T)$ such that $f^{-1}(t) = V(\gad(u))$, the statement holds.
Hence we may assume that for any node $t \in V(T)$, $|f^{-1}(t) \cap V(\gad(u))| \le 1$.

We build a~directed graph $\Aux$ whose underlying undirected graph is~$T$, as follows.
For any $x, y \in V(\Aux) := V(T)$, the arc $x \rightarrow y$ is in $E(\Aux)$ whenever $e := xy \in E(T)$ and, letting $(X_e, Y_e)$ be the cut induced by $e$ in $G^*$ with $f(X_e)$ (resp.~$f(Y_e)$) in the component of~$x$ (resp. of~$y$) in $T - e$, it holds that $Y_e$ contains a~copy of $P_u$.
Informally, the arcs of $\Aux$ point toward whole copies of~$P_u$.
We shall then simply show that are $x \neq y \in V(\Aux)$ such that both $x \rightarrow y$ and $y \rightarrow x$ are in $E(\Aux)$. 

By construction of $G^*$, there are $b > 4 \alpha \tau \ge \alpha |V(P_u)|$ copies of $P_u$: $P_u^1, \dots, P_u^b$.
Moreover, observe that for any $i, j, \ell$ with $i \neq j$ the $\ell$-th edges of $P_u^i$ and $P_u^j$ are mutually induced.
Hence \Cref{lem:cut-doesnt-cut} ensures that in any cut $(X, Y)$ of $G^*$, a~copy of $P_u$ is included in~$X$ or in~$Y$.
Thus each edge $xy \in E(T)$ implies that the arc $x \rightarrow y$ or the arc $y \rightarrow x$ is in $E(\Aux)$.

Assume, for the sake of contradiction, no edge $xy \in E(T)$ incurs that both $x \rightarrow y$ and $y \rightarrow x$ are in $E(\Aux)$.
It implies that $\Aux$ is an oriented tree, and thus contains a~sink (i.e., a~vertex with no outneighbors), say~$s$.
Since $T$ is of maximum degree 3, $s$ has at most three neighbors, say $v_1, v_2$ and $v_3$.
If the degree of $s$ is 2 (note that it cannot be less), some $v_i$ may not exist; in which case we set $v_i := s$.

Let us define $(V_1, V_2, V_3)$ the tripartition of $V(G^*) - f^{-1}(s)$ induced by $s$ as follows: $V_1$, $V_2$ and $V_3$ are the subsets of $V(G^*) - f^{-1}(s)$ mapped to the respective components of $v_1$, $v_2$ and $v_3$ in $T - s$ (with $V_i = \emptyset$ if and when $v_i = s$).
However, since $|f^{-1}(s) \cap V(\gad(u))| \leq 1$ there are $b - 1 > 6 \alpha \tau \geq \frac{3}{2} \alpha |V(P_u)|$ copies of $P_u$ lying in $V_1 \cup V_2 \cup V_3$.
Hence \Cref{lem:tricut-doesnt-cut} ensures that some copy of $P_u$ is included in one of $V_1, V_2, V_3$; a~contradiction to $s$ being a~sink.
\end{proof}

\newcommand{\group}{\text{group}}

We now define a \emph{grouping} operation on triples $(T, f, u)$, where $(T, f)$ is a subcubic \htree of $(G^*, \mathcal S^*)$ of sim-value smaller than $\frac{2(b-1)}{3\max_{u \in V(H)}|V(P_u)}$ (which is still very large compared to $\tau + \gamma$ by definition of $b$) and $u \in V(H)$, and denote it by $\group(T, f, u)$.
If there exists $t \in V(T)$ with $f^{-1}(t) = \gad(u)$, then we set $\group(T, f, u) := (T, f)$.
Otherwise, \Cref{lem:default-edge} ensures that there is an edge $e \in E(T)$ such that a copy of $P_u$ is in both sides of the cut defined by $e$.
We then define $\group(T, f, u) := (T', f')$, where
\begin{compactitem}
  \item $T'$ is obtained from $T$ by subdividing $e$, which adds a~node, say, $t_e$, and
  \item $f'$ satisfies that $f'(x) = f(x)$ whenever $x \not \in V(\gad(u))$, and $f'(x) = t_e$ otherwise.
\end{compactitem}
Given an edge $e' \in E(T')$, the edge \emph{corresponding to} $e'$ in $T$ is $e'$ if $e'$ is an edge of $T$, and $e$ if $e'$ is incident to $t_e$.

We make two observations on the grouping operation.

\begin{observation}\label{obs:subcubic}
For any subcubic \htree $(T, f)$ of $(G^*, \mathcal S^*)$ and any $u \in V(H)$, $\group(T, f, u)$ is a subcubic \htree of $(G^*, \mathcal S^*)$. 
\end{observation}

\begin{observation}\label{rmk:corresponding-cuts}
Let $(T, f)$ be a~subcubic \htree of $(G^*, \mathcal S^*)$ and $u \in V(H)$.
Let $(T', f') := \group(T, f, u)$, $e \in E(T)$, and $e' \in E(T')$ corresponds to~$e$ in~$T$.
Then if $(A_e, B_e)$ and $(A_{e'}, B_{e'})$ are the cuts of~$G^*$ defined by $e$ and $e'$, respectively, we have
\begin{compactitem}
	\item $A_{e'} \setminus V(\gad(u)) = A_e \setminus V(\gad(u))$, 
	\item $B_{e'} \setminus V(\gad(u)) = B_e \setminus V(\gad(u))$, 
	\item $V(\gad(u)) \subseteq A_{e'}$ implies that $A_e$ contains a copy of $P_u$, and 
        \item $V(\gad(u)) \subseteq B_{e'}$ implies that $B_e$ contains a copy of $P_u$.
\end{compactitem}
\end{observation}

We can next show that a~grouping can only decrease the sim-value. 

\begin{lemma}\label{lem:dec-sim-value}
For any subcubic \htree $(T, f)$ of $(G^*, \mathcal S^*)$ and any $u \in V(H)$, the~sim-value of $\group(T, f, u)$ is at~most that of $(T, f)$.
\end{lemma}
\begin{proof}
Let $(T', f') := \group(T, f, u)$.
Let $e'$ be an edge of $T'$ with $(A_{e'}, B_{e'})$ the cut of $G^*$ defined by $e'$.
Let $e \in E(T)$ be the edge corresponding to $e'$ in $T$ and $(A_e, B_e)$ be the cut of $G^*$ defined by $e$.
By construction of $(T', f')$, there is a node $t' \in V(T')$ with $f'^{-1}(t') = V(\gad(u))$.
Hence we can assume without loss of generality that $V(\gad(u)) \subseteq A_{e'}$.
Consider $M'$ an induced matching between $A_{e'}$ and $B_{e'}$.
We will prove that there exists an induced matching $M$ between $A_e$ and $B_e$ with $|M| = |M'|$.

Let $M' := \{a_1'b_1', \dots, a_p' b_p'\}$ with $a_i' \in A_{e'}$ and $b_i' \in B_{e'}$ for every $i \in [p]$.
We build the matching $M := \{a_1b_1, \dots, a_p, b_p\}$ as follows.
By assumption, $B_{e'} \cap V(\gad(u)) = \emptyset$, and by~\Cref{rmk:corresponding-cuts}, we know that $B_e \setminus V(\gad(u)) = B_{e'} \setminus V(\gad(u))$.
Hence $B_{e'} \subseteq B_e$, and so for any $i \in [p]$, we set $b_i := b_i'$.	
Exploiting the same idea, when $a_i' \not \in V(\gad(u))$ we have $a_i' \in A_e$, and so we set $a_i := a_i'$.
Otherwise we have $a_i' \in V(\gad(u))$, but since $V(\gad(u)) \subseteq A_{e'}$, \Cref{rmk:corresponding-cuts} ensures that some copy $P_u^k$ of $P_u$ is in $A_e$.
So we set $a_i$ to be the unique vertex in $\copies(a_i') \cap V(P_u^k)$.
It remains to prove that $M$ is indeed an induced matching.

Assume without loss of generality that $X' := \{a_1',\dots,a_t'\}$ are exactly the vertices in $\{a_1',\dots,a_p'\}$ that belongs to $V(\gad(u))$.
Let $X := \{a_1,\dots, a_t \}$.
Observe that by construction we have $V(M')\setminus X' = V(M) \setminus X$ and this set contains no vertex from $\gad(u)$.
For every $i\in [t]$, since $a_i \in \copies(a_i')$, we have $N(a_i) \setminus V(\gad(u)) = N(a_i') \setminus V(\gad(u))$.
In particular, $N(a_i) \cap (V(M)\setminus X) = N(a_i')\cap (V(M)\setminus X)$ and thus $a_i$ is adjacent to only $b_i=b_i'$ in $V(M)\setminus X$.
Since $X$ contains only the copies in $P_u^k$ of some vertices in $G$, we know that $X$ induces an independent set.
It follows that for each $i\in [t]$, $a_i$ is only adjacent to $b_i$ in $V(M)$.
As $V(M')\setminus X' = V(M) \setminus X$, we conclude that $M$ is an induced matching.

\end{proof}

We are now equipped to turn \htrees $G^*$ into \tmaps of $G^*$ no greater sim-value.


\begin{lemma}\label{lem:htree-to-tmap}
  If  $(G^*, \mathcal S^*)$ admits tree layout$(T,f)$ of sim-value at most $\alpha$, then $(G^*, \mathcal S^*)$ admits a~\tmap $(T',f')$ of sim-value at most the sim-value of $(T,f)$.
\end{lemma}
\begin{proof}
	The tree layout $(T, f)$ is a~subcubic \htree by definition.
	Let $(T_i, f_i)_{i \in [0,n]}$ be the sequences of hybrid trees where $(T_0,f_0):=(T,f)$, and for any $i \in [n]$, we set $(T_i, f_i) = \group(T_{i-1}, f_{i-1}, u_i)$ with $V(H) = \{u_1, \dots, u_n\}$.
	\Cref{lem:dec-sim-value} ensures that the sim-value of $(T_n,f_n)$ is at most the one of $(T, f)$.
	From the definition of the operation $\group$, it follows that for every node $t$ of $T_n$, we have either $f_n^{-1}(t) \in \mathcal S^*$ or $f_n^{-1}(t) = \emptyset$.
	
	Let $(T',f')$ be the \htree obtained by starting from $(T',f')= (T_n,f_n)$ and by doing the following.
	While there is an edge $tt'$ in $T'$ such that $f'^{-1}(t) \in \mathcal S$ and $f'^{-1}(t') = \emptyset$, we contract the edge $tt'$ into a~vertex whose preimage is the part $f'^{-1}(t)$.
	At every iteration, $(T',f')$ remains a \htree of $(G^*, \mathcal S^*)$.
	It can also be observed that after each iteration, the $\mathcal S^*$-cuts induced by the remaining edges of $T'$ doesn't change.
	Thus, by repeating this process, the sim-value can only decrease.
	
	At the end of this process, every node of~$T'$ has for preimage by~$f'$ a unique part of~$\mathcal S$.
	Hence $(T',f')$ is a~\tmap of~$(G^*,\mathcal S^*)$.
	By the argument in the previous paragraphs, its sim-value is at~most the sim-value of~$(T,f)$.
\end{proof}

We can conclude.

\begin{lemma}
  Let $(T, f)$ be a~tree layout witnessing that the sim-width of $G^*$ is at~most~$\tau+\gamma$.
  Then the tree sim-balancing of $(G,\mathcal S)$ is at~most $\tau+\gamma$.
\end{lemma}
\begin{proof}
  By \Cref{lem:htree-to-tmap}, there exists a~\tmap $(T',f')$ of $(G^*, \mathcal S^*)$ whose sim-value is at~most that of~$(T, f)$.
  We finally observe that $(T',g)$---with $g$ defined such that $g^{-1}(t) = S(u)$ whenever $f'^{-1}(t) = V(\gad(u))$---is a~\tmap of $(G,\mathcal S)$ of sim-value at~most~that of~$(T',f')$. 
\end{proof}

\end{document}